\documentclass[11pt,letterpaper]{article}
\usepackage[utf8]{inputenc}
\usepackage[english]{babel}
\usepackage{fullpage}
\usepackage[colorlinks=true, allcolors=blue]{hyperref}

\usepackage{adjustbox}
\usepackage{framed}
\renewenvironment{framed}[1][\hsize]
  {\MakeFramed{\hsize#1\advance\hsize-\width \FrameRestore}}%
  {\endMakeFramed}
\usepackage{float}
\usepackage{pgfplots}
\pgfplotsset{compat=1.15}
\usepackage{mathrsfs}
\usetikzlibrary{arrows}
\usepackage{lipsum}
\usepackage{footnote}
\usepackage{tikz-qtree}
\usepackage{amsmath}
\usepackage{algorithm}
\usepackage{comment}
\usepackage{algpseudocode}
\usepackage{tikz}
\usetikzlibrary{trees}
\usepackage{amsthm}
\usepackage{amssymb}
\usepackage{nicefrac}
\newcommand{\PP}{\text{I\kern-0.15em P}}
\usepackage{graphicx}
\renewcommand{\eqref}[1]{\hyperref[#1]{(\ref*{#1})}}

\author{Arghya Chakraborty \and Rahul Vaze}

\newtheorem{theorem}{Theorem}[section]

\newtheorem{lemma}[theorem]{Lemma}
\newtheorem{proposition}[theorem]{Proposition}
\newtheorem{claim}[theorem]{Claim}
\newtheorem{definition}[theorem]{Definition}
\newtheorem{remark}[theorem]{Remark}

\newtheorem{property}[theorem]{Property}
\usepackage[capitalize,nameinlink]{cleveref}
\definecolor{dtsfsf}{rgb}{0.8274509803921568,0.1843137254901961,0.1843137254901961}
\definecolor{qqqqff}{rgb}{0,0,1}
\definecolor{qqffqq}{rgb}{0,1,0}
\definecolor{ffqqqq}{rgb}{1,0,0}

\title{Online facility location with weights and congestion}
\begin{document}

\maketitle

\begin{abstract}


The classic online facility location problem deals with finding the optimal set of facilities in an online fashion when demand requests arrive one at a time and facilities need to be opened to service these requests. In this work, we study two variants of the online facility location problem; (1) weighted requests and (2) congestion. Both of these variants are motivated by their applications to real life scenarios and the previously known results on online facility location cannot be directly adapted to analyse them.

\begin{description}

\item[Weighted requests:] In this variant, each demand request is a pair $(x,w)$ where $x$ is the standard location of the demand while $w$ is the corresponding weight of the request. The cost of servicing request $(x,w)$ at facility $F$ is $w\cdot d(x,F)$. For this variant, given $n$ requests, we present an online algorithm attaining a competitive ratio of $\mathcal{O}(\log n)$ in the secretarial model for the weighted requests and show that it is optimal. 

\item[Congestion:] The congestion variant considers the case when there is an additional congestion cost that grows with the number of requests served by each facility. For this variant, when the congestion cost is a monomial, we show that there exists an algorithm attaining a constant competitive ratio. This constant is a function of the exponent of the monomial and the facility opening cost but independent of the number of requests.

\end{description}


\end{abstract}

\newpage

\section{Introduction}

The facility location problem is one of the most well-studied problems in the field of algorithms \cite{vaziranifl,arora,lin_vitter}. Meyerson \cite{959917} studied the online version of this problem where requests arrive in a sequence and on arrival, the request must be allocated to a facility and a cost is incurred depending on the distance of the request from the facility. New facilities can also be opened on the arrival of a new request so that the current request and future requests have a nearby facility to be allocated to. However opening facilities incur a cost too and the goal is to minimize the total cost incurred.

To motivate the variants of online facility location we study, let us consider a plausible real-life scenario where such an online facility location algorithm may be deployed. Suppose the government is trying to set up vaccination centres in order to distribute vaccines depending on incoming requests. While online facility location can model the problem, a few other issues arise.  For instance, one facility might not be able to satisfy a large number of requests and depending on the number of requests, there will be congestion. Also, there may be multiple vaccine doses resulting in weighted requests where each request needs to be satisfied a certain number of times (for instance, some people might require 2 doses while others require more\footnote{We assume that these multiple requests must be served by the same facility else these multiple requests may be treated as independent requests.}). There are numerous other real life applications of the facility location problem. These additional variants involving congestion and weighted-requests have their own costs which need to be taken into account when designing a solution and the classical facility location problem may not be directly applicable in these cases. In this work, we generalize the classical facility location problem to incorporate these extra variants and prove matching upper and lower bounds. 
We study two variants of the classical facility location, each incorporating one of these variants. We now define the classical online facility location problem and these two variants formally.

\subsection{Online Facility Location Problem}
The facility location problem comprises of a set of requests $\{x_1,x_2,...,x_n\}$ on a metric space $(\mathscr{M},d(\cdot\;,\;\cdot))$. The objective is to open a set of facilities $F$ on this metric space. There is a facility opening cost of $f$ incurred for opening each such facility and a distance cost for each request which is proportional to the distance of the request to its nearest facility. Naturally, the objective is to open facilities such that the total cost incurred is minimum, where the total cost incurred is given by
\[f\cdot|F|+\sum_{i=1}^nd(x_i,F)\]
where $d(x_i,F)$ denotes the distance of $x_i$ to its nearest facility in $F$. This problem is known to be NP-Hard \cite{FOWLER1981133}.

The online variant of the facility location problem was first studied in detail by Meyerson \cite{959917}. In this setting, a request, $x_i$, needs to be allotted \emph{on its arrival} to a facility, $F$, either an existing facility or a newly-opened facility. The cost incurred on allotting the request is the distance cost $d(x_i,F^{(i)})$ where $F^{(i)}$ refers to the set of facilities available at the time of servicing request $x_i$ (including possibly the newly-opened facility). As in the offline setting, opening new facilities costs $f$ for every facility opened. Meyerson \cite{959917} constructed a randomized algorithm, $RFL$ (Randomized Facility Location), which attains a competitive ratio of $\mathcal{O}(\log n)$ in expectation. Later, Fotakis \cite{Fotakis2008} proved that no online algorithm can attain a competitive ratio better than $\Omega\left(\frac{\log n}{\log \log n}\right)$ and also did an improved analysis proving that $RFL$ attains a competitive ratio of $\mathcal{O}\left(\frac{\log n}{\log \log n}\right)$ in expectation. 

Meyerson also showed that in the secretarial model (where the requests are adversarial but the arrival order is uniformly at random), his algorithm attains a competitive ratio of $8$.

\subsection{The two variants and our results}
Online facility location problem has many applications in real world. However, as observed by the vaccination example, some of these applications may have additional constraints on the classical online facility problem. In this paper, we study two variants incorporating these additional constraints.

\subsubsection{Variant I: Weighted Requests}

In our first variant, we consider a setting of weighted requests, where each request $x_i$ arrives with an additional weight $w_i$. One may view this weight $w_i$ as the number of times the request $x_i$ needs to be served or simply as a weight corresponding to a premium client. 

Formally, a sequence of requests $(x_i,w_i)$ arrive in an online fashion where $x_i$ is a point in the metric space $(\mathscr{M},d(\cdot\;,\;\cdot))$ while $w_i$ is a positive real number. On arrival, the request $(x_i,w_i)$ must be allocated to a facility, which could be one of the existing facilities or a newly-opened facility. As before, the facility opening cost is $f$ but the cost incurred on allocating request $(x_i,w_i)$  is $d(x_i,F^{(i)})\cdot w_i$, instead of the usual $d(x_i,F^{(i)})$. If $w_i\cdot d(x_i,\cdot)$ is also a metric, this can be handled by the classical case, but this may not always be the case. As usual, the goal is to minimize the total cost incurred which is given by the following expression:
\[f.|F|+\sum_{i=1}^n d(x_i,F)\cdot w_i\]

Note, that while our motivating example suggested that the weights $w_i$'s are positive integers (corresponding to the number of times a request needs to be served), the above formulation is more general and allows for any positive real number.

We first modify Fotakis' lower bound to show that in the worst-case setting, no online algorithm can attain a competitive ratio better than ${\Omega}(n)$ matching the na\"{i}ve algorithm that opens a new facility at every request. We then show that in the secretarial setting (when the requests are adversarial, but the ordering is random), there is an online algorithm that achieves a competitive ratio of $\mathcal{O}(\log n)$, which we show is tight once again by adapting Fotakis' lower bound. Observe that these results prove that this variant is provably different from the classical setting.

\subsubsection{Variant II: Congestion}
Our next variant incorporates the notion of congestion: if a facility is attending to multiple requests, there is an additional cost that it has to pay, depending on the number of requests it is attending to. 

Formally, a sequence of requests $x_i$ arrives and on arrival, the request must be allocated to a facility. Just as before, a facility may be opened in order to allocate $x_i$. A distance cost of $d(x_i,F^{(i)})$ is incurred and a facility opening cost of $f$ is incurred every time a facility is opened. As of yet, this is exactly same as the online facility location problem. However we incorporate an additional congestion cost. If the facility to which $x_i$ is allocated to has $k$ requests allotted to it after $x_i$ is allotted, then a congestion cost of $g(k)-g(k-1)$ is incurred in addition, where $g$ is a convex non-decreasing function, satisfying $g(0)=0$. We shall call $g$ as the congestion function. In this model, if a total of $\ell$ requests are allocated to a facility, a total congestion cost of $g(\ell)$ is incurred and the goal is to minimize the total cost incurred. Note that the total cost includes three components: facility-opening costs, allocation costs and congestion costs.

In this work, we prove results for the special case when $g$ is a monomial. We show that Meyerson's RFL can be modified to obtain a $\mathcal{O}(\log k^*/\log \log k^*)$ competitive ratio,  where $k^* = 2\cdot g^{-1}\left(\frac{f}{g(2)-2}\right)$. Note that this is independent of $n$, the number of requests. We also show that this is tight up to a constant by providing a matching lower bound.

A few words on the restriction of congestion function to monomials. Our results work for any convex non-decreasing congestion function $g$ that satisfies $g(0) =0$ and $g(a\cdot b)=g(a)\cdot g(b)$ for all $a,b$. Hence, the assumption that $g$ is a monomial. While the assumption that the congestion function $g$ is a monomial is restrictive, it is interesting that in this setting, we get tight matching upper and lower bounds. It would be interesting to extend this proof to more general convex functions.



\subsubsection{Proof Techniques}
While all the algorithms that we analyse are modifications of Meyerson's algorithm $RFL$ and the lower bound analyses are inspired by Fotakis' analysis, significant modifications are needed to handle each variant. 

For the weighted requests, as indicated earlier, we show that in the adversarial model, no algorithm can attain a competitive ratio better than $\Omega(n)$, which is trivially achievable by opening a facility for all the requests. Then we analyse this in the secretarial model, where we are able to connect this problem to a problem on increasing sequences in a uniformly random permutation. We obtain both the upper and lower bound of $\Theta(\log n)$ by studying properties of increasing sequences in random permutations. In particular, we reduce the analysis of the algorithm to the analysis of a process related to permutations, which we refer to as the \emph{Selection Process}, which might be of independent interest. The Selection Process has the property that if the input permutation is chosen adversarially, its cost can be as high as $\Omega(n)$ while on a random permutation, the expected cost drops to $\Theta(\log n)$. 

For the model with congestion, we first show that the optimal offline algorithm ($OPT$) will not allocate more than $k^*$ requests to a facility. Using this fact, we analyze any arbitrary facility $c^*$ opened by $OPT$. We analyze the cost incurred by our online algorithm over the requests allocated to $c^*$. We are able to split up this cost carefully into two components -- one component is simply the congestion cost while the other component has no congestion cost involved. Using the fact that the facility $c^*$ has at most $k^*$ requests allotted to it, we are able to attain a competitive ratio of $\mathcal{O}\left(\frac{\log k^*}{\log\log k^*}\right)$.

\subsection{Related Work}
Since many real-life problems are similar to facility location problems, it has received a lot of attention. It has even been used as a tool to solve other online problems (like \cite{1198387,10.1145}) . The survey by Fotakis\cite{fotsurvey} gives a nice depiction of the problem and the general techniques that one needs to be familiar with. Other modifications that have been studied, include a setting where the facilities can themselves be moved with some cost \cite{mobfac}, or the requests may be deleted \cite{del}. The latter paper also studies capacitated facility location, which assumes that each facility can attend to at most a fixed number of requests. This is very similar to our congestion model, but as the bounds suggest, these two problems seem to be different.

Similarly, there has been interest on making these algorithms deterministic. Anagnostopoulos, Bent, Upfal, and Hentenryck in their paper \cite{bentupfal} show that there is an $\mathcal{O}(\log n)$-competitive deterministic algorithm for online facility location. Later Fotakis in \cite{fotsurvey} showed a deterministic primal dual algorithm attaining  $\mathcal{O}\left(\frac{\log n}{\log \log n} \right)$ competitive ratio. More recently Kaplan, Naori, and Raz in \cite{recent} show that when the requests arrive uniformly at random, no online algorithm can achieve a competitive ratio better than $2$ and a modification of Meyerson's algorithm achieves a competitive ratio of $3$ in expectation. 

Another interesting topic in online algorithms is to improve the performance based on advice. This has been studied widely in \cite{advice,learning,learning2,learning3} and indeed one can have algorithms with better competitive ratio if the advice is good.

One thing to note is that in all of these papers, the algorithms in most cases are modifications of the original algorithm by Meyerson. The analysis is what makes these problems interesting. This is true for our results as well where the algorithms are adaptations of Meyerson's algorithm however the analysis needs some novel ideas.


\section{Facility Location with Weighted Requests}\label{sec:time}

In this model, we assume that input requests are ordered pairs of the form $(x_i,w_i)$ where $x_i$ describes the position of the request while $w_i$ is the weight of the request. This could be the model for how long the request stays in the system or the number of times the request has to be serviced. The corresponding cost incurred is $w_i\cdot d(x_i,F)$ (as opposed to just $d(x_i,F)$). Our results hold for arbitrary positive weight units, which need not be integers.

\begin{definition}[Weighted Online Facility Location]
A sequence of $n$ ordered pairs $(x_i,w_i)$ are given as input, where $x_i\in \mathscr{M}$ for some metric space $\mathscr{M}$. Each such request needs to be allocated to a facility in an online fashion. Opening a facility incurs a cost of $f$ and allocating a request $(x_i,w_i)$ to a facility $F$ incurs a cost of $d(x_i,F)\cdot w_i$. The goal is to open facilities and allocate the requests in an online fashion such that the total cost incurred is minimized.
\end{definition}
Here, the algorithm knows the metric space $\mathscr{M}$ but does not know $n$ beforehand.

It is not hard to show that in the worst case setting, one can not expect to achieve a competitive ratio better than $\Omega(n)$ (See \cref{lbwc} for more details), which is trivially achievable by opening a facility on every request.

In light of the above, we study this setting in the secretarial model. Informally speaking, the secretarial model is one where the requests may be adversarial but the order in which they appear is uniformly at random. This ensures that the arrival order of the input requests cannot be adversarial. In other words, while the $n$ input requests may be arbitrary, the secretarial model assumes that at each step all the remaining requests are equally likely to arrive as the next request. For the secretarial model, we show the following matching upper and lower bounds :

\begin{theorem}\label{secmodel}
In the online facility location problem with weighted requests, no online algorithm can obtain a competitive ratio better than $\Omega(\log n)$ in the secretarial model.
\end{theorem}

\begin{theorem}\label{thmtrfl}
For the online facility location problem with weighted requests, there exists an algorithm attaining a competitive ratio of $\mathscr{O}(\log n)$ in expectation under the secretarial model.
\end{theorem}

The lower bound is deferred to \cref{sec:timed_lb} and the upper bound is proved in \cref{ubss}.

\section{Lower bounds for the weighted-request variants}\label{sec:timed_lb}

For the worst case setting, we shall work with a particular sequence of input requests. Then for the secretarial case, we shall use the same set of input requests assuming that the order of arrival of these requests is uniformly at random. Hence, we will first describe this set of requests.
\subsection{The input for the lower bounds}

The input requests will be constructed using a binary tree on a star metric. This is similar to the lower bound presented by Fotakis \cite{Fotakis2008}.

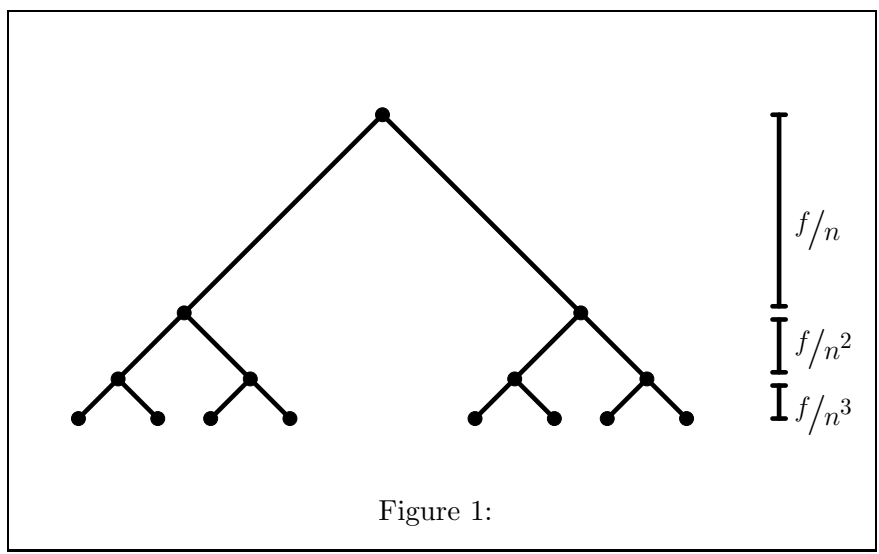
\begin{figure}[!ht]
\centering
\begin{framed}[0.7\textwidth]
\begin{adjustbox}{width=\textwidth}

\begin{tikzpicture}[line cap=round,line join=round,>=triangle 45,x=1cm,y=1cm]
\clip(-5.3,-5.2) rectangle (7.1,1.2);
\draw [line width=2pt] (0,0)-- (-3,-3);
\draw [line width=2pt] (-3,-3)-- (-4,-4);
\draw [line width=2pt] (-3,-3)-- (-2,-4);
\draw [line width=2pt] (-4,-4)-- (-4.6,-4.6);
\draw [line width=2pt] (-4,-4)-- (-3.4,-4.6);
\draw [line width=2pt] (-2,-4)-- (-2.6,-4.6);
\draw [line width=2pt] (-2,-4)-- (-1.4,-4.6);
\draw [line width=2pt] (0,0)-- (3,-3);
\draw [line width=2pt] (3,-3)-- (2,-4);
\draw [line width=2pt] (2,-4)-- (1.4,-4.6);
\draw [line width=2pt] (2,-4)-- (2.6,-4.6);
\draw [line width=2pt] (3,-3)-- (4,-4);
\draw [line width=2pt] (4,-4)-- (3.4,-4.6);
\draw [line width=2pt] (4,-4)-- (4.6,-4.6);
\draw [line width=2pt] (6,0)-- (6,-2.9);
\draw [line width=2pt] (6,-3.1)-- (6,-3.9);
\draw [line width=2pt] (6,-4.1)-- (6,-4.6);
\draw (6.1,-1.3) node[anchor=north west] {\LARGE$\nicefrac{f}{n}$};
\draw (6.1,-4.1) node[anchor=north west] {\LARGE$\nicefrac{f}{n^3}$};
\draw (6.1,-3.1) node[anchor=north west] {\LARGE$\nicefrac{f}{n^2}$};
\draw [line width=2pt] (5.9,0)-- (6.1,0);
\draw [line width=2pt] (5.9,-2.9)-- (6.1,-2.9);
\draw [line width=2pt] (5.9,-3.1)-- (6.1,-3.1);
\draw [line width=2pt] (5.9,-3.9)-- (6.1,-3.9);
\draw [line width=2pt] (5.9,-4.1)-- (6.1,-4.1);
\draw [line width=2pt] (5.9,-4.6)-- (6.1,-4.6);
\begin{scriptsize}
\draw [fill=black] (0,0) circle (3pt);
\draw [fill=black] (-3,-3) circle (3pt);
\draw [fill=black] (3,-3) circle (3pt);
\draw [fill=black] (-4,-4) circle (3pt);
\draw [fill=black] (-2,-4) circle (3pt);
\draw [fill=black] (2,-4) circle (3pt);
\draw [fill=black] (4,-4) circle (3pt);
\draw [fill=black] (-4.6,-4.6) circle (3pt);
\draw [fill=black] (-3.4,-4.6) circle (3pt);
\draw [fill=black] (-2.6,-4.6) circle (3pt);
\draw [fill=black] (-1.4,-4.6) circle (3pt);
\draw [fill=black] (1.4,-4.6) circle (3pt);
\draw [fill=black] (2.6,-4.6) circle (3pt);
\draw [fill=black] (3.4,-4.6) circle (3pt);
\draw [fill=black] (4.6,-4.6) circle (3pt);
\end{scriptsize}
\end{tikzpicture}
\end{adjustbox}\caption{}\label{bintree}
\end{framed}
\end{figure}

There will be requests along the nodes of the tree\footnote{The metric used for this analysis is the shortest path metric on the binary tree. However this can be embedded on the Euclidean Metric over $\mathbb{R}^1$. This was shown in \cite{Fotakis2008}.} in \cref{bintree}, forming a path, starting from the root node and all the way to a leaf node. In order to refer to this tree we shall use the term ``level''. The root node is at level $0$, the children of the root node are level $1$ and for any node, its level is $1$ added to the level of its parent node, all the way to the leaves being level $n-1$ nodes.

The tree will be such that the distance between consecutive nodes will keep decreasing: the distance between a node at level $i$ and a node at level $i+1$ is $\frac{f}{n^{i+1}}$. We shall have $1$ request at each node: the request at the $i$-th level node will have weight $n^i$.

For the worst case setting the requests will arrive in order of their levels - starting from the root node all the way to the leaf node. For the secretarial setting, the requests will be as mentioned but the arrival order will be uniformly at random, as we have mentioned.
\subsection{Lower Bound in the Worst Case Setting}\label{lbwc}

\begin{proposition}
No online algorithm can attain a competitive ratio better than $\Omega(n)$.
\end{proposition}
\begin{proof}
For the mentioned input requests, an offline algorithm $OFF$ may open $1$ facility at the leaf node and in which case the distance of the root node to facility is at most $2\frac{f}{n}$ and for a node in the $i$-th level, its distance will be at most $2\frac{f}{n^{i+1}}$ from the leaf node. However taking into account the weight, the cost for the $i$-th level node is at most $2\frac{f}{n^{i+1}}. n^{i}=2\frac{f}{n}$  Hence the total cost paid by $OFF$ is less than $f+\sum_{i=0}^{n-1}2\frac{f}{n}=f+2f=3f$.

Let us now consider the performance of an online algorithm on this input. At every level the algorithm must either
\begin{itemize}
    \item Open a new facility
    \item Or pay a distance cost (taking into account the weight of the request) to its parent node at the very least
\end{itemize}
This is because, when opening a new facility, the online algorithm may either open a facility on the node of the current request (say $v_i$) or try to guess future nodes and open a facility somewhere in the subtree of the current level node request. Both of these will result in a new facility nonetheless. Now if the algorithm tries to guess and open a future node in the subtree (let us say left sub tree of $v_i$) instead of opening a facility at the current node, the adversary will select the next node for input request from the other subtree (right subtree of $v_i$) thus adversarially ensuring that guessing never gives the algorithm a correct future node.

This ensures that the algorithm does not have any available facility in the subtree of the current request. Hence, if the algorithm decides to pay a distance cost, on an $i$-th level request, instead of opening a facility anywhere, it must pay a distance cost at least to its parent, $\frac{f}{n^{i}}$ with weight $n^{i}$ resulting in a cost of $f$. The other case consists of opening a facility and in that case the algorithm incurs a cost of $f$ anyway. Therefore the algorithm incurs a cost of $f$ for each level of the tree, resulting in a total cost of $nf$.

One needs to also consider the case where the algorithm opens multiple facilities on a requested node ensuring that both the subtrees have facilities in them. However in this case the algorithm has already paid the facility opening cost multiple times. If the algorithm opens $k$ facilities on a requested $i$-th level node node (for $k\geq 2$), the algorithm cannot ensure that all the nodes in the $i+(k-1)$-th level below have facilities in their subtrees (since there are $2^{k-1}$ nodes at that level and $2^{k-1}\geq k$). This ensures that opening more than one facility on one node request is not beneficial for the algorithm.

Now let us compare the cost of the optimal offline algorithm, $OPT$, to the cost of an online algorithm. As we have already seen, the cost of any online algorithm is at least $fn$ but the cost of $OPT$ is at most $3f$ ensuring that $\frac{C_{A}}{C_{OPT}}\geq \frac{fn}{3f}=\frac{n}{3}$ for any online algorithm $A$.

Notice that this proof works against deterministic algorithms only. In order to prove this for randomized algorithms, we may use Yao's principle where as input we select one path among the $2^{n-1}$ paths uniformly at random. Alternatively, one can do this using the same technique (\cref{guess}, in particular) used in \cref{secmodel}, hence we have skipped the proof over here.
\end{proof}

A competitive ratio of $n$ is trivially achievable by opening a facility on every request. However now, the naturally interesting question is whether an online algorithm can perform better if the input is secretarial.

\subsection{Lower Bound in the Secretarial Setting}\label{lbss}
\begin{definition}
In the secretarial model, we assume that the adversary decides the input requests but not their arrival order. Instead after the adversary has decided the set of input requests, the arrival order of the requests is uniformly at random.
\end{definition}

Now, we shall give a proof for the lower bound in this setting.

\begin{proof}[Proof of {\cref{secmodel}}]
As mentioned earlier, we shall provide the same set of input requests but now the arrival order will be uniformly at random.

Notice that this makes the problem easier for an online algorithm in the sense that if an $i$-th level node is the first request to arrive, all the requests of level $0,1,...,i-1$ may be deduced by the online algorithm from the $i$-th level node and hence it will not have to open facilities for those requests when they arrive. Therefore we will not try to lower bound the cost of the algorithm on a request at level $i$ if another request of higher level had already arrived earlier. However every time a node is requested such that no node below it has been requested yet, we shall ensure that the online algorithm has to pay some cost.

First we shall look at online algorithms, $A$, such that if it opens a facility on a request, it will open only at the location of the current request.

\begin{claim}\label{noguessalgo}
Let $A$ be an online algorithm such that on input requests, $A$ would either open a facility at the request location or not open a facility at all. Then $A$ cannot have a competitive ratio better than $\Omega(\log n)$.
\end{claim}
\begin{proof}[Proof of Claim]
As discussed earlier, the input will be composed in the following fashion :
\begin{itemize}
    \item For an input of size $n$, we shall consider a binary tree on $n$ levels : $T_n$.
    \item Then we shall choose a leaf node uniformly at random among the $2^{n-1}$ leaf nodes.
    \item There is a unique $n$ length path connecting the root node to this leaf node.
    \item The input will comprise of requests on these nodes such that a request on the $i$-th level node will have weight of $n^i$.
    \item Also the distance between the $i-1$-th level node and the $i$-th level node will be $\frac{f}{n^i}$.
\end{itemize}
Notice that since we are in the secretarial model, these requested nodes (with corresponding weights) will arrive uniformly at random. Let us suppose that there are $k$ requests $v_1,v_2,...,v_k$ such that on arrival of the request $v_i$, it is the highest levelled node (or in other words, the lowest placed node in the tree structure) that has arrived till then.

Then for each of those requests, $v_i$, the online algorithm must either open a facility there or pay a distance cost with weight. Also, we have the guarantee that the nearest open facility available to the requested node is at least as far away as the distance to its parent node and hence when multiplied by the weight, the cost is exactly $f$, which is the same cost incurred for opening a facility. Therefore every time a request arrives such that no node of higher level has arrived yet, the online algorithm has to pay a cost of $f$. To finish off our proof we just need the following lemma :
\begin{lemma}\label{greedyinc}
Let $S:=[n]$ be the set containing the numbers $1$ through $n$. Let $\pi$ be a uniformly at random permutation of the set $S$. Let $k_\pi$ be the number of times that the $i$-th element in the permutation $\pi$ is the largest element observed till then. Then $\mathbb{E}_\pi[k_\pi]=\Theta(\log n)$.
\end{lemma}
\begin{proof}
Let $X_i$ be a random variable such that
\[   
X_i = 
     \begin{cases}
       1 &\quad\text{if $\pi_i>\pi_j$ for all $j<i$}\\
       $0$ &\quad\text{otherwise.} \\ 
     \end{cases}
\]
Notice that $\mathcal{P}[X_n=1]=\frac{1}{n}$ since $X_n$ will be $1$ if and only if $\pi_n=n$ which has probability $\frac{1}{n}$. Also, conditioned on $\pi_{i+1},\pi_{i+2},...,\pi_n$, the probability that $X_i=1$ is $\frac{1}{i}$ since no matter what $\pi_{i+1},...,\pi_n$ are, $\pi_i$ will be the largest among $\pi_1,\pi_2,...\pi_i$ with probability $\frac{1}{i}$. Hence $\mathbb{E}[k_\pi]=\mathbb{E}[\sum X_i]=\sum\mathbb{E}[X_i]=\sum\frac{1}{i}=\Theta(\log n)$.
\end{proof}

This \cref{greedyinc} gives us that in the secretarial input model, in expectation there will be $\Theta(\log n)$ requested nodes such that they were the highest levelled nodes on their arrival. Hence cost of online algorithm $A$ would be $\Omega (f\cdot\log n)$ compared to the offline optimal algorithm whose cost is at most $3f$ as we had calculated earlier (Here we are using the fact that the cost of the offline algorithm is the same for the secretarial input model as was for the non-secretarial model). Therefore the competitive ratio of online algorithm $A$ is $\Omega(\log n)$.
\end{proof}
Now we consider the case that an online algorithm may not only open facilities at requested nodes but it may also open a facility at a nearby location instead of opening the facility exactly on the request. This potentially allows the online algorithm to guess future nodes and maybe reduce the cost. We shall look at the same input that we worked with, in the \cref{noguessalgo}, and observe that such a guessing algorithm cannot do much better in expectation.

\begin{lemma}\label{guess}
Let $T_n$ be a binary tree with a path chosen uniformly at random from the root node to a leaf node. If a node, $v$, of this path is revealed, and an algorithm $A$ opens a facility in the subtree of $v$, the guess will match with the actual path on $2$ nodes (apart from $v$) in expectation.
\end{lemma}
\begin{proof}
Since the path was chosen uniformly at random, $A$ can guess the next level node with probability $1/2$. Also, $2$ nodes can be guessed correctly by $A$ with probability $1/4$ and continuing in this fashion, the probability that $A$'s guess matches with the path on $i$ nodes is at most $\frac{1}{2^i}$.

Therefore the expected number of nodes that $A$ can guess correctly by opening a facility is at most 
$$S:=\sum_{i=1}^\infty \frac{i}{2^i}$$
Firstly, we can use Ratio Test to observe that this series is actually converging. Then we notice that $$2S=\sum_{i=1}^\infty \frac{i}{2^{i-1}}=\sum_{i=0}^\infty \frac{i+1}{2^i}=1+\sum_{i=1}^\infty \frac{i+1}{2^i}=1+\sum_{i=1}^\infty \frac{i}{2^i}+\sum_{i=1}^\infty \frac{1}{2^i}=1+S+1$$
$$\implies S=2\;.$$
\end{proof}

This means that every time the algorithm $A$ receives a node which is the highest level till then, it has the choice to pay a distance cost with weight, or open a facility at the request itself or opening a facility in the subtree of the node, guessing future nodes. We have already seen the competitive ratio for algorithms that do not guess future nodes. Let $k$ be the number of nodes observed such that they were the highest levelled nodes seen on their arrival. Then we had seen that an algorithm that does not guess incurs a cost of $f$ each time resulting in a cost of $k\cdot f$ and then we had observed that the expected value of $k$ is of the order of $\log n$ to obtain the final competitive ratio.

Now however the algorithm may guess the future nodes and potentially incur a less cost. Let $k$ be the number of nodes observed such that they were the highest levelled nodes seen on their arrival. Then we can show that $A$ must incur a cost of $f\cdot\frac{k}{3}$ at least, in expectation. If $A$ incurs a cost less than $f\cdot\frac{k}{3}$, then algorithm $A$ has opened $<\frac{k}{3}$ facilities, and hence has $<\frac{k}{3}$ guesses. Therefore the total number of nodes that $A$ has guessed correctly is $<2\cdot\frac{k}{3}$ in expectation (Using \cref{guess}). Note that on the arrival of node, say $v$, even if the algorithm $A$ might guess a node, say $v_1$, correctly (Where $v_1$ lies in the subtree of $v$), it may not reduce the cost incurred by $A$. This is because the next highest levelled node to arrive might be a child of $v_1$ in which case $A$ needs to pay a cost of $f$ anyway. In other words, correctly guessing a node might not reduce the number of highest levelled nodes observed, because the nodes guessed correctly might not be the highest levelled nodes to arrive. Therefore, taking into consideration the $<\frac{k}{3}$ nodes that $A$ has paid for and the nodes that $A$ has potentially guessed correctly, $A$ has satisfied $<\frac{k}{3}+\frac{2k}{3}=k$ nodes in expectation (even after considering facilities opened, correctly guessed nodes and distance costs paid). This cannot be the case since $A$ has to satisfy all the $k$ nodes which were the highest levelled nodes on arrival.

Therefore $A$ must incur a cost of $f\cdot\frac{k}{3}$ at least, in expectation. Since $k=\Theta(\log n)$ in expectation (Using \cref{greedyinc}), $A$ must incur a cost of $\Omega(\log n)$ in expectation.

\end{proof}

\section{Upper Bound for Weighted-Requests in the Secretarial Setting}\label{ubss}

Now we shall focus our attention on the setting with weighted requests and present the proof of \cref{thmtrfl}.

In order to show this we will reduce the problem to a completely different problem, which we shall call \textit{The Selection Process}. This is a novel idea and connecting the Weighted Facility Location Problem to the Selection Process is what allows us to get a tight upper bound of $\mathcal{O}(\log n)$. 

\begin{definition}[The Selection Process]\label{selprob}
Let $S=[n]$. A sequence of $(p_i,j_i)$ arrive in an online fashion where $p_i\in[0,1]^n$ is a vector of probabilities and $j_i\in S$ for $i\in[n]$. The first pair to arrive is $(p_1,j_1)$ and we shall \textbf{select} the element $j_1$ with probability $p_{1,j_1}$.

Similarly, on input $(p_i,j_i)$, we shall select the element $j_i$ with probability $p_{i,j_i}$ but only if the element $j_i$ is greater than the previously selected numbers.

Also, for all $1\leq i\leq n-1$, and $j_i\in S$, we shall assume that $p_{i,j_i}\geq p_{i+1,j_i}$. The final quantity that we want to compute is the expected number of selected elements.
\end{definition}

Let us first consider the case when the adversary selects both the $p_i$'s and the $j_i$'s. In this setting one can see that in the worst case, there  can be as many as $n$ elements selected. For this, consider the probability vectors $p_i$'s to be all $1$'s throughout and the element arrival order be $j_1=1,j_2=2,...,j_n=n$. In this case, all the elements will be selected. 

Hence, let us consider the case when the adversary decides the probability vectors $p_i$'s (satisfying the constraint of $p_{i,j_i}\geq p_{i+1,j_i}$) however the arrival order of the requests is uniformly at random. That is, we take a permutation, $\pi$ of $[n]$ uniformly at random and at step $i$, the element appearing is $\pi_i$. In this secretarial setting, one can show that no matter what probabilities the adversary chooses, the number of selected elements is $\mathscr{O}(\log n)$ in expectation.

\begin{claim}\label{assumeclaim}
The number of selected elements in The Selection Process is $\mathscr{O}(\log n)$ in expectation.
\end{claim}

We shall prove \cref{assumeclaim} later.

\begin{remark}
Notice that the selection process does not necessarily choose the longest increasing subsequence since the longest increasing subsequence in a uniformly permuted array, is of the size of $\Theta(\sqrt{n})$. It is actually enough to observe that this is $\Omega(n)$, which follows from Erd\H{o}s–Szekeres theorem \cite{erdos}.

On the other hand, \cref{greedyinc} tells us that greedily selecting the largest element gives us an increasing subsequence of size $\Theta(\log n)$.

However the adversary may manipulate the probability vectors $p_i$'s in a manner to make the largest increasing subsequences to be more likely to be picked. Indeed, in \cref{whyvectareincr}, we show that if the probability vectors aren't enforced to be non-increasing, the adversary can make sure that there are $\Theta(\sqrt n)$ selections in expectation. This is why the \cref{assumeclaim} is of utmost importance.
\end{remark}

\begin{remark}\label{whyvectareincr}
The fact that for all $j\in S$, and $1\leq i\leq n-1$, we require $p_{i,j}\geq p_{i+1,j}$ is necessary. One may look at the example where the first $\sqrt{n}$ probability vectors are such that they have $1$ in their first $\sqrt{n}$ terms and $0$ everywhere else. The next $\sqrt{n}$ probability vectors have their  first $\sqrt{n}$ terms $0$, next $\sqrt{n}$ terms as $1$ and all $0$'s after that. This continues for $\sqrt{n}$ many blocks of probability vectors, each with their corresponding block of $1$'s and remaining $0$'s. It is easy to see that given these probability vectors and a uniform arrival order of elements, one can expect $\Theta(\sqrt{n})$ elements to be selected in The Selection Process and $\mathcal{O}(\log n)$ would have been unachievable.
\end{remark}

First, let us present our algorithm for facility location with weighted requests. Then we shall analyse the competitive ratio and see where The Selection Process comes in.

\begin{algorithm}[H]
  \caption{Weighted Randomized Facility Location (WRFL)}\label{WRFL}
  \begin{algorithmic}[1]
    \Procedure{$WRFL$}{$x_i$,$w_i$}\Comment{Input request is $x_i$ with weight $w_i$}
      \State $F_i\gets \texttt{Nearest facility to }x_i$
      \State $p_i\gets \min\{1,\frac{d(x_i,F_i)\times w_i}{f}\}$
      \State With probability $p_i$ open a new facility at $x_i$ and allocate $x_i$ to it
      \State With probability $1-p_i$, assign $x_i$ to $F_i$
    \EndProcedure
  \end{algorithmic}
\end{algorithm}

While algorithm $WRFL$ is a simple modification of \cref{RFL} (Meyerson's Algorithm $RFL$, which we describe later), the analysis is completely different. The analysis of $RFL$ does not go through because of the weight component. This is also apparent from the fact that the competitive ratios are of different order : $\log n$ for $WRFL$ while $\frac{\log n}{\log \log n}$ for $RFL$. Hence the entire structure of the analysis is different.

Now, we shall prove \cref{thmtrfl} for the \cref{WRFL} in particular.

\begin{proof}[Proof of {\cref{thmtrfl}}]
In order to analyse the algorithm, we shall divide the run of the algorithm into what we shall define as phases. However before that, we shall set up some notations and observe a few properties.

Let $c^*$ be a facility opened by $OPT$ and let $S_{c^*}$ be the set of requests that $OPT$ assigns to $c^*$. Also when we say that an algorithm incurs a cost of $C$ over a subset of requests, say $T$, we will look at each request in the subset $T$ and note how much the algorithm paid for that request when allocating it- be it distance cost or facility opening cost. For example, if the algorithm allocates the request to a facility which was already opened on a previous request, not from $T$, then we just take into account the distance cost without considering the facility opening cost for that facility.

We may focus our attention on one facility $c^*$, opened by $OPT$ and look at the set of requests $S_{c^*}$ that $OPT$ assigns to $c^*$. Using \cref{clm1}, we can say that analysing one such cluster is sufficient for our analysis. For the sake of convenience, we shall refer to this set of requests as $S$, whenever we mean $S_{c^*}$.

Now, in order to analyze this, we shall define the phases of the algorithm.
\begin{definition}
Corresponding to the facility $c^*$, opened by $OPT$, the run of the algorithm $WRFL$ can be divided into phases as follows :
\begin{itemize}
    \item The algorithm is said to be in Phase $0$ until $WRFL$ opens a facility, say $F_1$, on a request in $S$, at which point phase $1$ starts.
    \item Phase $1$ starts when $WRFL$ opens its first facility on a request of $S$ (which we called $F_1$) and continues until $WRFL$ opens another facility, say $F_2$ on a request of $S$ and this facility is closer to $c^*$ than $F_1$.
    
    \hspace{2cm}\vdots
    \item Phase $i$ starts when $WRFL$ opens a facility $F_i$ and continues until $WRFL$ opens another facility, say $F_{i+1}$, on a request of $S$ and this facility, $F_{i+1}$, is closer to $c^*$ than $F_i$.
    
\end{itemize}
\end{definition}

Observe that this definition of Phase is valid for the given cluster of requests $S$. With respect to another cluster, the phases would be different. Also, the algorithm does not know about the phases, it is only well defined during the analysis of the algorithm. Just to get a clear idea, we can note that the number of phases depends not only on which requests $WRFL$ opens a facility on, but it also depends on the arrival order of the requests.

\begin{figure}[!ht]
\begin{framed}

\definecolor{ffqqqq}{rgb}{1,0,0}
\definecolor{rvwvcq}{rgb}{0.08235294117647059,0.396078431372549,0.7529411764705882}
\begin{adjustbox}{width=0.9\textwidth}
\begin{tikzpicture}[line cap=round,line join=round,>=triangle 45,x=1cm,y=1cm]

\clip(-1.1234957010748252,-4.909014115479848) rectangle (17.014579996386516,4.774774523513391);
\draw [shift={(0,0)},line width=2pt]  plot[domain=-0.5235987755982991:0.5235987755982988,variable=\t]({1*5*cos(\t r)+0*5*sin(\t r)},{0*5*cos(\t r)+1*5*sin(\t r)});
\draw [shift={(0,0)},line width=2pt]  plot[domain=-0.5235987755982991:0.5235987755982988,variable=\t]({1*9*cos(\t r)+0*9*sin(\t r)},{0*9*cos(\t r)+1*9*sin(\t r)});
\begin{scriptsize}
\draw [fill=rvwvcq] (0,0) circle (2.5pt);
\draw[color=rvwvcq] (0.37,0.20) node {\LARGE\LARGE$c^*$};
\draw [fill=rvwvcq] (8.37,3.3) circle (2.5pt);
\draw[color=rvwvcq] (8.7,3.3
) node {\LARGE$F_{i}$};
\draw [fill=rvwvcq] (4.7,-1.7) circle (2.5pt);
\draw[color=rvwvcq] (5.4
,-1.8) node {\LARGE$F_{i+1}$};
\draw [fill=ffqqqq] (3.1,0.4) circle (2.5pt);
\draw [fill=ffqqqq] (7.5,-3.2) circle (2.5pt);
\draw [fill=ffqqqq] (5.1,0.8) circle (2.5pt);
\draw [fill=ffqqqq] (10.2,2.7) circle (2.5pt);
\draw [fill=ffqqqq] (10.5,-1.2) circle (2.5pt);
\draw [fill=ffqqqq] (12.7,-4.1) circle (2.5pt);
\draw [fill=ffqqqq] (13.8,2.7) circle (2.5pt);
\end{scriptsize}
\end{tikzpicture}

\end{adjustbox}\caption{Phase $i$ starts when Facility $F_i$ is opened and ends when Facility $F_{i+1}$ is opened where by definition $F_{i+1}$ should be closer to $c^*$ than $F_i$. All the requests in Phase $i$ are shown in red dots. They may be closer or farther away from $F_{i}$ with respect to $c^*$. The only guarantee that we have is that the first request closer than $F_i$ that results in a facility opening would change the phase to $i+1$.}
\end{framed}
\end{figure}

\begin{lemma}\label{lemma1}
Consider a set of requests $S$ assigned to the same facility by $OPT$. Assume that the request arrivals and $WRFL$'s random choices are such that there are $k$ phases, by our definition. Then the cost ratio, for the requests in $S$, of $WRFL$ to $OPT$ is $\mathscr{O}(k)$ in expectation.
\end{lemma}
\begin{proof}
Let us first compute the cost of $OPT$ over the requests in $S$. We note first that $OPT$ has only opened one facility for these requests, at $c^*$. This constitutes a cost of $f$. Other than this, $OPT$ incurs a cost of $C:=\sum_{x_i\in S}d(x_i,c^*)\cdot w_i$\;.

To complete the proof, we shall use two properties in order to analyse the cost of $WRFL$.
\begin{property}\label{prop1}
Let $X = \{(x_1,w_1),(x_2,w_2),\ldots, (x_m,w_m)\}$ be a subsequence of $m$ requests. Then $WRFL$, in expectation, pays a cost $\leq f$ before it opens the first new facility in $X$.

(This holds even if there are several open facilities at the beginning and even if new facilities are opened between two requests in the subsequence)
\end{property}
\begin{property}\label{prop2}
On input request $(x_i,w_i)$, let $F$ be the nearest facility to $x_i$. Then the cost that $WRFL$ pays in expectation for the request $(x_i,w_i)$ is at most $2d(x_i,F)\cdot w_i$.
\end{property}

The proofs of these two Properties have been studied in earlier works (e.g. \cite{Fotakis2008,fotsurvey}), but for the sake of completion we have added them in \cref{sec:misproof}.

Using \cref{prop1}, algorithm $WRFL$ would incur a cost of $f$ in expectation before it opens a facility on a request of $S$. This would result in phase changing from Phase $0$ to Phase $1$. Now let us say that we are at Phase $i$ and we want to estimate $WRFL$'s cost over the requests in Phase $i$, where Phase $i$ started when $WRFL$ opened a facility $F_i$. Here $F_i$ is closer than $F_{i-1}$ to $c^*$, by construction.

\textit{(Type $1$ requests)} Let $S_1$ be the set of requests such that $d(c^*,x)<d(c^*,F_i)$ for all $x\in S_1$. Then using \cref{prop1}, algorithm $WRFL$ pays a cost of $f$ in expectation over requests in $S_1$ before it opens a facility $F_{i+1}$ on a request with $d(c^*,F_{i+1})<d(c^*,F_i)$. This event of opening a facility ends up changing the phase from $i$ to $i+1$. Hence for phase $i$, requests closer to $c^*$ than $F_i$ incur a cost of at most $f$ in expectation.

\textit{(Type $2$ requests)} Let $S_2$ be the set of requests such that $d(c^*,x)\geq d(c^*,F_i)$ for all $x\in S_2$. Now observe that $OPT$ incurs a  cost of $d(c^*,x)\cdot w_x$ for the request $x$ (Where $w_x$ denotes the weight of the request $x$). Let us now try to estimate the cost incurred by $WRFL$ for this request $x$. Here, $WRFL$ has an open facility $F_i$ and $d(F_i,x)\leq d(x,c^*)+d(c^*,F_i)$ using triangle inequality. Also, $d(c^*,x)\geq d(c^*,F_i)$ by assumption. So $d(F_i,x)\leq d(x,c^*)+d(x,c^*)=2d(x,c^*)$. Therefore the nearest facility to $x$ for $WRFL$ is at most at a distance of $2d(x,c^*)$. Hence by \cref{prop2}, the cost incurred by $WRFL$ for the request $x$ is at most $4d(x,c^*)\times w_x$ in expectation, which is $4$ times the cost incurred by $OPT$ itself.

Hence the total cost of Type $2$ requests for $WRFL$ is at most $4C$ in expectation (Where $OPT$ had incurred a cost of $f+C$ for the requests in $S$). On the other hand, $WRFL$ incurs a cost of $f$ for Type $1$ requests in expectation, for each phase. Also $WRFL$ incurs a facility opening cost of $f$ for each phase. Therefore the total cost incurred by $WRFL$ is at most $2k\cdot f+4C$ in expectation. This results in an expected competitive ratio of $2k+4=\mathscr{O}(k)$.

\end{proof}
\begin{lemma}\label{lemma2}
The expected number of phases is $\mathscr{O}(\log n)$.
\end{lemma}
\begin{proof}
In order to view this we will reduce the problem to The Selection Process (\cref{selprob}), which we had mentioned earlier.

We shall also assume the \cref{assumeclaim} for now and finish the proof of \cref{lemma2}, after which we shall give a formal proof of \cref{assumeclaim}.  Let the requests in $S$ be ordered in decreasing distance from $c^*$ and labelled accordingly. Thus $x_1$ is the request farthest from $c^*$, $x_2$ is the second farthest request and so on till $x_n$ which is the request nearest to $c^*$.

Just before the first request from $S$ arrives online, there may be open facilities due to other requests. Given the configuration of the facilities (opened by $WRFL$) at this stage, each request $x_i\in S$ has a certain probability of opening a facility if it were the first request to arrive in $S$. This vector of probabilities will be our $p_1$ with $p_{1,i}$ being the probability for $x_i$ opening a facility, if it was the first request to arrive. Now one of the requests $x_{\pi_1}$, arrives uniformly at random. The index of the request arriving corresponds to the element arriving in The Selection Process. Note that the vector of probabilities thus generated, depends on the previous requests and whether these requests have opened new facilities or not. However, this still adheres to the Selection Process, as long as the requests arrive uniformly at random.

At any stage when the $i$-th request $x_{\pi_i}$ arrives, we can similarly compute the vector $p_i$ and also this request would change the phase only if a facility is opened there and also if it is nearer to $c^*$ than all the other opened facilities, which corresponds to the index of the request being larger than the indices of all the opened facilities. So, any time the phase changes, a facility is opened which is nearer to $c^*$ than all previously opened facilities. This corresponds to the element being selected in the Selection Process.

Also while these probability vectors may be arbitrary, the probability of opening a facility on a particular request can only decrease over time. This is because if the request arrives later on, it can potentially have nearby facilities, reducing the probability of opening a facility but this probability of opening a facility cannot increase with time, implying $p_{i,j}\geq p_{i+1,j}$.

Therefore, if we can show that the number of selected elements in The Selection Process is $\mathscr{O}(\log n)$ in expectation, then we would have shown that the number of phases is $\mathscr{O}(\log n)$ in expectation, as was required.

\end{proof}

Using \cref{lemma1,lemma2}, it follows that the competitive ratio of $WRFL$ is $\mathscr{O}(\log n)$.

\end{proof}

\begin{proof}[Proof of {\cref{assumeclaim}}]
In order to prove this, we shall reduce the problem one more time.

Notice that the $i$-th element to arrive is $\pi_i$ and is selected with probability $p_{i,\pi_i}$, if no larger element has yet been selected. We can view the selection probability of $\pi_i$ as $$p_{i,\pi_i}=p_{1,\pi_i}\times\frac{p_{2,\pi_i}}{p_{1,\pi_i}}\times\frac{p_{3,\pi_i}}{p_{2,\pi_i}}\times...\times\frac{p_{i,\pi_i}}{p_{i-1,\pi_i}}\;.$$
Since $p_{j,\ell}\geq p_{j+1,\ell}$ for all $j$ (by definition), each of the terms $\frac{p_{j+1,\pi_i}}{p_{j,\pi_i}}$ is less than or equal to $1$. Hence we can view the terms, $\frac{p_{j+1,\pi_i}}{p_{j,\pi_i}}$, as probabilities.

Now just before the first element $\pi_1$ arrives, we can toss random coins for each $i$ such that the $i$-th coin is $1$ with probability $p_{1,i}$. Given that $\pi_1$ is the first element to arrive, we can look at the realization of $p_{1,\pi_1}$ and select it if it is $1$.

Now just before the second element arrives, we can toss random coins for each $i$ such that the $i$-th coin is $1$ with probability $\frac{p_{2,i}}{p_{1,i}}$. Now, after observing $\pi_2$, the second element to arrive, we can look at the realizations of $p_{1,\pi_2}$ and $\frac{p_{2,\pi_2}}{p_{1,\pi_2}}$ and select the element $\pi_2$ only if both the entries are $1$, and additionally if $\pi_2>\pi_1$. Notice that $\pi_2$ is selected with probability $p_{2,\pi_2}$ in this manner, as needed.

Similarly this process continues for the $i$-th element arriving, $\pi_i$, for all $i\geq 3$. Here we toss random coins for each element $i$ that are $1$ with probability $\frac{p_{i,\pi_i}}{p_{i-1,\pi_i}}$ and we take into account the realizations of the previous coin tosses of probability $\frac{p_{i',\pi_i}}{p_{i'-1,\pi_i}}$ (for all $i'<i$). The element $\pi_i$ is selected only if all of the realizations are $1$ and $\pi_i$ is larger than all selected elements. One can see that selecting the elements in this manner produces the same probability of an element being selected.

However now we can state that an element $\pi_i$ that has not yet arrived will certainly not be selected on its arrival if even one of the realizations $\frac{p_{i',\pi_i}}{p_{i'-1,\pi_i}}$ or $p_{1,\pi_i}$ happens to be $0$ (Where $i'<i$, since $\pi_i$ has not yet arrived), also we can say that the element $\pi_i$ will certainly not be selected on its arrival if an element larger than $\pi_i$ has already been selected.

Let us take one such array $A[.][.]$ of realizations of the probabilities. Where $A[i][j]$ is $1$ with probability $\frac{p_{i,j}}{p_{i-1,j}}$. On the arrival of an element, $\pi_i$, we look at the $\pi_i$-th column and select the element $\pi_i$ only if all the first $i$ entries of the column are $1$. Hence once a $0$ appears in a column, that element will certainly not be selected and instead of tossing further coins for the element, we shall set all the remaining entries in the column to be $0$. This maneuver does not alter the probability of any arriving element being selected. Also if a larger element had already been selected, the corresponding element will certainly not be selected.
\begin{figure}[!ht]
\begin{framed}[0.7\textwidth]
\centering
\begin{adjustbox}{width=\textwidth}
\begin{tikzpicture}[line cap=round,line join=round,>=triangle 45,x=1cm,y=1cm]
\clip(-3.5,-16.9) rectangle (18,0.9);

\draw [line width=2pt] (0,0)-- (16,0);
\draw [line width=2pt] (16,0)-- (16,-16);
\draw [line width=2pt] (16,-16)-- (0,-16);
\draw [line width=2pt] (0,-16)-- (0,0);
\draw [line width=2pt] (2,0)-- (2,-16);
\draw [line width=2pt] (4,0)-- (4,-16);
\draw [line width=2pt] (6,0)-- (6,-16);
\draw [line width=2pt] (8,0)-- (8,-16);
\draw [line width=2pt] (14,0)-- (14,-16);
\draw [line width=2pt] (0,-2)-- (16,-2);
\draw [line width=2pt] (0,-4)-- (16,-4);
\draw [line width=2pt] (0,-14)-- (16,-14);
\draw [line width=2pt] (0,-6)-- (16,-6);
\draw [line width=2pt] (0,-8)-- (16,-8);
\draw (-1.5,-0.6) node[anchor=north west] {\Huge$p_1$};
\node[label=below:\rotatebox{-90}{\Huge{Outcome of coin toss of $i$th iteration$\longrightarrow$}}] at (-2.5,-0.2) {};
\draw (-1.5,-2.46) node[anchor=north west] {\Huge$\frac{p_2}{p_1}$};
\draw (-1.5,-4.46) node[anchor=north west] {\Huge$\frac{p_3}{p_2}$};
\draw (-1.5,-6.46) node[anchor=north west] {\Huge$\frac{p_4}{p_3}$};
\draw [line width=2pt] (10,0)-- (10,-16);
\draw [line width=2pt] (12,0)-- (12,-16);
\draw [line width=2pt] (0,-10)-- (16,-10);
\draw [line width=2pt] (0,-12)-- (16,-12);
\draw (-1.5,-8.46) node[anchor=north west] {\Huge$\frac{p_5}{p_4}$};
\draw (-1.9,-14.46) node[anchor=north west] {\Huge$\frac{p_n}{p_{n-1}}$};
\draw (5.94,-10.3) node[anchor=north west] {\Huge$\frac{p_{i,j}}{p_{i-1,j}}$};
\draw (0.75,1) node[anchor=north west] {\Huge{Element $j\longrightarrow$}};
\draw (2.75,-0.6) node[anchor=north west] {\LARGE$1$};
\draw (4.75,-0.6) node[anchor=north west] {\LARGE$1$};
\draw (6.75,-0.6) node[anchor=north west] {\LARGE$1$};
\draw (0.75,-0.6) node[anchor=north west] {\LARGE$1$};
\draw (2.75,-2.6) node[anchor=north west] {\LARGE$1$};
\draw (4.75,-2.6) node[anchor=north west] {\LARGE$1$};
\draw (2.75,-4.6) node[anchor=north west] {\LARGE$1$};
\draw (10.75,-2.6) node[anchor=north west] {\LARGE$1$};
\draw (10.75,-0.6) node[anchor=north west] {\LARGE$1$};
\draw (14.75,-0.6) node[anchor=north west] {\LARGE$1$};
\draw (12.75,-0.6) node[anchor=north west] {\LARGE$1$};
\draw (14.75,-2.6) node[anchor=north west] {\LARGE$1$};
\draw (14.75,-4.6) node[anchor=north west] {\LARGE$1$};
\draw (14.75,-6.6) node[anchor=north west] {\LARGE$1$};
\draw (6.75,-2.6) node[anchor=north west] {\LARGE$1$};
\draw (6.75,-4.6) node[anchor=north west] {\LARGE$1$};
\draw (2.75,-6.6) node[anchor=north west] {\LARGE$1$};
\draw (2.75,-8.6) node[anchor=north west] {\LARGE$1$};
\draw (0.75,-2.6) node[anchor=north west] {\LARGE$0$};
\draw (2.75,-10.6) node[anchor=north west] {\LARGE$0$};
\draw (4.75,-4.6) node[anchor=north west] {\LARGE$0$};
\draw (6.75,-6.6) node[anchor=north west] {\LARGE$0$};
\draw (8.75,-0.6) node[anchor=north west] {\LARGE$0$};
\draw (10.75,-4.6) node[anchor=north west] {\LARGE$0$};
\draw (12.75,-2.6) node[anchor=north west] {\LARGE$0$};
\draw (14.75,-8.6) node[anchor=north west] {\LARGE$0$};
\end{tikzpicture}
\end{adjustbox}\caption{Array $A$ of probabilities}
\end{framed}
\end{figure}
\begin{claim}\label{imdone}
Let $i$ be the first index when selection happens and $\pi_i$ be the first element to be selected. Then either $\mathbb{E}[\pi_i]\geq \frac{n}{3}$ or the $i$-th row of $A$ has at least $\frac{n}{3}$ entries as $0$.
\end{claim}
\begin{proof}
Let the $i$-th row of $A$ have at least $\frac{2n}{3}$ entries as $1$. This means that the columns corresponding to those $\frac{2n}{3}$ entries have $1$ throughout, for the first $i$ entries. We shall then conclude that $\mathbb{E}[\pi_i]\geq \frac{n}{3}$ to complete this proof.

Let $c_1,c_2,...,c_k$ be the columns of $A$ with $1$'s throughout till the $i$-th row. By our assumption, $k\geq \frac{2n}{3}$. Under the condition that the first element to be selected is the $i$-th element, the expected value of the selected element equals
$$\mathbb{E}[\pi_i]\stackrel{(1)}{=}\frac{1}{k}\sum _{j=1}^k c_j\stackrel{(2)}{\geq} \frac{1}{k}\sum _{j=1}^k j=\frac{k+1}{2}\geq \frac{\frac{2n}{3}+1}{2}\geq \frac{n}{3}$$
Equality $(1)$ follows from the fact that because of the secretarial model, all the elements $c_j$ are equally likely to be the $i$-th element. Since the $i$-th element, $\pi_i$, is the first selected element, $\pi_i$ has to be one of the elements $\{c_1,c_2,...,c_k\}$ only. Also none of the elements $c_j$ have appeared earlier, else they would have been selected first. Hence all of those $k$ elements are equally likely to be $\pi_i$.

Inequality $(2)$ holds because all the $k$ elements are distinct integers from $[n]$. Therefore the sum will be the smallest when the $k$ elements are $1,2,...,k$.

Hence we have that if the $i$-th row of $A$ has at least $\frac{2n}{3}$ entries as $1$ then $\mathbb{E}[\pi_i]\geq \frac{n}{3}$, which proves our claim.
\end{proof}
Applying \cref{imdone} recursively completes the proof. When the first element is selected, in expectation there will be $\frac{n}{3}$ elements that cannot be selected - either because they are smaller than the element picked or because they already have $0$'s in their columns in $A$. Therefore we can \textit{discard} all those elements that can no longer be selected - meaning, if element $j$ can no longer be selected we may remove the $j$-th row and the $j$-th column from the matrix $A$. Now we will have a smaller matrix, $A_1[.][.]$, leaving apart the rows and columns of the discarded elements. On this matrix, $A_1$, we shall again apply the \cref{imdone} to discard $\frac{1}{3}$ fraction of its elements in expectation, when the next element is selected. This process continues and \cref{imdone} states that each time an element is selected, $\frac{1}{3}$rd of the remaining elements are discarded in expectation. Hence this process ends in $\mathscr{O}(\log n)$ many steps in expectation. In other words, the number of elements selected is $\mathscr{O}(\log n)$.

\end{proof}

\section{Facility Location with Congestion}\label{sec:congestion_full}

In this model, we consider a sequence of $n$ requests $x_1,...,x_n$ from a metric space $(\mathscr{M},d(\cdot,\cdot))$ endowed with a distance metric $d$. The goal is to open facilities when needed and assign these incoming requests to the facilities. We assume we are given a congestion function $g$, which is a convex non-decreasing function such that $g(0)=0$.  Opening a facility incurs a cost of $f$, while assigning a request $x_i$ to a facility $F$ containing $k$ requests before $x_i$ was assigned incurs a cost of $d(x_i,F)+g(k+1)-g(k)$. The additional $g(k+1)-g(k)$ is the congestion cost at the facility due to the new request.

\begin{definition}[Online Facility Location with Congestion]
A sequence of $n$ requests, $x_i$, is given as input, where $x_i\in \mathscr{M}$ for some metric space $\mathscr{M}$. Each such request needs to be allocated to a facility in an online fashion. Opening a facility incurs a cost of $f$ and allocating a request $x_i$ to a facility $F$ containing $k$ requests before $x_i$ was assigned incurs a cost of $d(x_i,F)+g(k+1)-g(k)$. The goal is to open facilities and allocate the requests in an online fashion such that the cost incurred is minimized.
\end{definition}

In this paper, we consider congestion costs which satisfy $g(a\cdot b) = g(a) \cdot g(b)$ (for example a monomial). For this particular model, we prove the following :

\begin{theorem}\label{thm1}
\begin{enumerate}
\item In the online facility location problem with congestion, there exists an online algorithm attaining a competitive ratio of $\mathcal{O}(\frac{\log k^*}{\log\log k^*})$ in expectation, where $k^* := 2\cdot g^{-1}\left(\frac{f}{g(2)-2}\right)$, a constant independent of the number of requests.\label{enu:cong_ub}
\item 
Furthermore, no randomized online algorithm can achieve a competitive ratio better than $\frac{\log k^*}{\log \log k^*}$.\label{enu:cong_lb} 
\end{enumerate}
\end{theorem}
In the rest of this section, we first prove the Lower Bound for the Worst Case input in \cref{lb}. Then we prove that the same bound is achieved, in expectation, by a randomized algorithm in \cref{ub}. However before we do these, we shall make a key observation that in the congestion model, an offline optimal algorithm will not allocate too many requests to a single facility.

Let $k'$ be the smallest integer such that $g(k'+1)-g(k')\geq f$. If an online algorithm allocates $k'+1$ requests to one facility then the congestion cost for the $k'+1$-th request itself is $g(k'+1)-g(k')\geq f$. So no online algorithm needs to put more than $k'$ requests in a facility as it might as well open a new facility and incur less cost.
\par\noindent
The above reasoning holds for offline algorithms too and specifically for the offline optimal algorithm, $OPT$. However we can give an even better bound when considering OPT. For this we shall assume that the congestion function satisfies $g(a\cdot b)=g(a)\cdot g(b)$.
\begin{claim}\label{opt-bound}
Any facility opened by OPT has at most $k^*:=2\cdot g^{-1}\left(\frac{f}{g(2)-2}\right)$ requests allotted to it.
\end{claim}
\begin{proof}
Let $F$ be a facility opened by OPT with $k$ requests in it. Since OPT did not split the $k$ requests into two facilities but kept them in a single facility, it must be the case that splitting them into two facilities costs more. Notice that the cost paid if two facilities are opened is $\geq 2f+g(a_1)+g(a_2)$ in addition to some distance cost, where $a_1+a_2=k$. Since $g$ is a convex non-decreasing function, we have $g(a_1)+g(a_2)\geq 2g\left(\frac{k}{2}\right)$ (Here we have assumed that $k$ is even, else we can look at $k-1$ instead of $k$). Also opening two facilities means that the total distance cost is lesser compared to opening only one facility. Hence we can conclude 
$$2f+2g\left(\frac{k}{2}\right)\geq f+g(k)\;.$$
Since we are working under the assumption that $g(a\cdot b)=g(a)\cdot g(b)$, we can now write this as 
$$f\geq (g(2)-2)g\left(\frac{k}{2}\right)$$
$$\implies \frac{k}{2}\leq g^{-1}\left(\frac{f}{g(2)-2}\right)$$
$$\implies k\leq 2.g^{-1}\left(\frac{f}{g(2)-2}\right)\;.$$
\end{proof}
\subsection{Lower Bound}\label{lb}
Here we shall present a proof of \cref{thm1}-\cref{enu:cong_lb}. The proof will be along the same lines as the previous lower bound proof.
\begin{proof}[Proof of {\cref{thm1}-\cref{enu:cong_lb}}]
Consider a binary tree of depth $h$. The distance between root node and its children is $\frac{f}{h}$ but the distance keeps decreasing as we traverse down the tree. The distance between any level $i$ node and its child will be $\frac{f}{h^{i+1}}$, where we will decide later what $h$ to choose. The input will comprise of $1$ request at root node, followed by $h$ requests at a node at level $1$, and so on, with $h^i$ requests for some node at the $i$th level. The requested nodes will form a path from the root node to a leaf node, say $\ell$. Also there will be a total of $n=1+h+h^2+...+h^{h-1}=\frac{h^h-1}{h-1}$ requests.

\begin{figure}[ht]
\begin{framed}[0.9\textwidth]
\label{tree}
\begin{adjustbox}{width=\textwidth}
\begin{tikzpicture}[line cap=round,line join=round,>=triangle 45,x=1cm,y=1cm]
\clip(-7.786097219966367,-2.0857989746113774) rectangle (8.447811873179946,5.659128488577184);
\draw [line width=2pt] (0,5)-- (-2,2);
\draw [line width=2pt] (0,5)-- (-2,2);
\draw [line width=2pt] (0,5)-- (2,2);
\draw [line width=2pt] (-2,2)-- (-3,0);
\draw [line width=2pt] (-2,2)-- (-1,0);
\draw [line width=2pt] (2,2)-- (1,0);
\draw [line width=2pt] (2,2)-- (3,0);
\draw [line width=2pt] (-3,0)-- (-3.54,-0.97);
\draw [line width=2pt] (-3,0)-- (-2.32,-0.99);
\draw [line width=2pt] (-1,0)-- (-1.58,-0.99);
\draw [line width=2pt] (-1,0)-- (-0.36,-0.99);
\draw [line width=2pt] (1,0)-- (0.42,-0.99);
\draw [line width=2pt] (1,0)-- (1.62,-0.99);
\draw [line width=2pt] (3,0)-- (2.4,-0.97);
\draw [line width=2pt] (3,0)-- (3.64,-0.99);
\begin{scriptsize}
\draw [fill=dtsfsf] (0,5) circle (2.5pt);
\draw [fill=dtsfsf] (-2,2) circle (2.5pt);
\draw [fill=black] (2,2) circle (2.5pt);
\draw [fill=black] (-3,0) circle (2.5pt);
\draw [fill=dtsfsf] (-1,0) circle (2.5pt);
\draw [fill=black] (1,0) circle (2.5pt);
\draw [fill=black] (3,0) circle (2.5pt);
\draw [fill=black] (-3.54,-0.97) circle (2.5pt);
\draw [fill=black] (-2.32,-0.99) circle (2.5pt);
\draw [fill=dtsfsf] (-1.58,-0.99) circle (2.5pt);
\draw [fill=black] (-0.36,-0.99) circle (2.5pt);
\draw [fill=black] (0.42,-0.99) circle (2.5pt);
\draw [fill=black] (1.62,-0.99) circle (2.5pt);
\draw [fill=black] (2.4,-0.97) circle (2.5pt);
\draw [fill=black] (3.64,-0.99) circle (2.5pt);
\draw (-1.8,4) node[anchor=north west] {\LARGE $\frac{f}{h}$};
\draw (-1.4,1.5) node[anchor=north west] {\LARGE $\frac{f}{h^2}$};
\draw (-2,0) node[anchor=north west] {\Large $\frac{f}{h^3}$};
\end{scriptsize}
\end{tikzpicture}
\end{adjustbox}
\caption{The sequence of input requests will consist of $1$ request at the root node, followed by $h$ requests at the next node, $h^2$ requests at the next node and so on. These nodes where requests will be, are shown in red and form a path from the root node to a leaf node.}

\end{framed}
\end{figure}

Notice that since any offline algorithm knows this entire sequence of requests, it can open one facility at the appropriate leaf node $\ell$. This would result in a facility opening cost of $f$ and a congestion cost of $g(n)$ and a distance cost which we will compute next. The root node will be at a distance of $\frac{f}{h}+\frac{f}{h^2}+\frac{f}{h^3}+...+\frac{f}{h^{h-1}}\leq 2\frac{f}{h}$ from the leaf node and as stated earlier, it will have one request so it will have to pay $\leq 2 \frac{f}{h}$ distance cost once. The next node at level $1$ will have to pay the distance cost of $\frac{f}{h^2}+\frac{f}{h^3}+...+\frac{f}{h^{h-1}}\leq 2\frac{f}{h^2}$ but it will have $h$ requests so it will have to pay this cost $h$ times. The $i$-th level node will have to pay a distance cost of $\leq 2\frac{f}{h^{i+1}}$ but it will have to pay it $h^i$ times. Since $OPT$ does not allocate more than $k^*$ requests to a facility, we shall select our input such that the total number of requests are $k^*$ or less, in order to achieve the best lower bound.

Therefore total distance cost $\leq 2\cdot \frac{f}{h}\cdot (h-1)\leq 2\cdot f$ and hence $C_{OPT}\leq f+g(k^*)+2\cdot f\;.$

Now any online algorithm on receiving a request at a level for the first time, can decide to do one of the following : 
\begin{itemize}
    \item Open a new facility at the request in the level
    \item Use another previously opened facility and pay a distance cost
    \item Try to guess the path following which the future requests would arrive and open a facility at a later level
\end{itemize}
Since we are assuming an adaptive adversarial model where the adversary can decide the next request depending on the previous actions of the algorithm, if any online algorithm tries to guess the path and opens a later level facility the adversary will simply choose another path so that the path from this newly opened facility to future requests will be via the parent node, $x$. So for the current level requests the algorithm will have to pay a distance cost at least to its parent node and future level requests will also have to pass through the parent node, $x$, since there is no open facility in the subtree of $x$ other than $x$ by construction. Therefore, the distance cost to future level requests will not be reduced either. Hence trying to guess future nodes will certainly not provide a better cost for online algorithms.

Now if the algorithm decides to open a facility at each level, its cost would be $f.(h+1)$ and there would be no distance cost. However even if the algorithm tries to not open a facility at a node level and use one of the earlier facilities, it will have to pay a distance cost. If the algorithm does not open a facility at the $i$-th level node, the algorithm will at least have to pay a distance cost of $\frac{f}{h^i}$ since this is the distance to its parent node and whatever open facility the algorithm will try to use to allocate this request will have to be through the parent node. Also, since there will be $h^i$ requests at level $i$, the total distance cost incurred in this level would be $\geq h^i.\frac{f}{h^i}= f$.

Now competitive ratio of any online algorithm $A$ would be $\mu_A=\frac{C_A}{C_{OPT}}$ where $C_A$ is the cost $A$. As we have seen earlier, $\frac{C_A}{C_{OPT}}\geq \frac{f.(h+1)}{f+g(k^*)+2f}\;.$

$$\therefore{\mu_A \geq \frac{f.(h+1)}{3f+g(k^*)}}=\frac{f.(h+1)}{3f+f\frac{g(2)}{g(2)-2}}=\frac{h+1}{3+\frac{g(2)}{g(2)-2}}=\Omega(h)\;.$$
However $h$ cannot be arbitrarily large since the total number of requests equals $1+h+h^2+...+h^h$ can be at most $k^*$, by construction.

Now $1+h+h^2+...+h^h=h^h+(1+h+h^2+...+h^{h-1})=h^h+\frac{h^h-1}{h-1}=\Omega(h^h)\leq k^*\;.$

Equivalently, we can say that $h^h\leq \mathcal{O}(k^*)$, allowing us to state that :
\begin{remark}\label{hpowerh}
The competitive ratio of any algorithm solving the online facility location problem with congestion is $\mu_A\geq \Omega(h)$, where $h$ is such that $h^h\leq \mathcal{O}(k^*)$.
\end{remark}
Notice that the larger that $h$ is, the better is the lower bound that we obtain and the only constraint preventing us from selecting an arbitrarily large $h$ is the fact that $h^h\leq \mathcal{O}(k^*)$. Hence in order to attain the best bound possible, we must select $h$ such that $h^h=\Theta(k^*)$, or $h=\Theta(\frac{\log k^*}{\log \log k^*})$ using our notation. This gives us a lower bound of $\frac{\log k^*}{\log \log k^*}$ for the competitive ratio of any algorithm solving the online facility location problem with congestion.
\end{proof}

Notice that we have proved this for deterministic algorithms only. One may use \cref{guess} to extend this proof for randomized algorithms just as we had done earlier.
\subsection{Upper Bound}\label{ub}
We have already seen a lower bound of the online facility location problem with congestion. Now we would like to find an online algorithm to solve the online facility location problem with congestion and study the upper bound. We shall present an algorithm that asymptotically attains the same competitive ratio as suggested by the lower bound, making our analysis tight.

We shall present this algorithm as a modified version of $RFL$. Then we shall study the modified algorithm and try to compare the competitive ratio of the modified algorithm with congestion to the competitive ratio of the original algorithm without congestion.

However notice that the lower bounds for the competitive ratio is different for the case with congestion as compared to the original case without congestion. As we have seen earlier, the lower bound for the competitive ratio is $\frac{\log k^*}{\log \log k^*}$, when we are looking at online facility location with congestion. However in the original setting without congestion, the lower bound was $\frac{\log n}{\log \log n}$, which as we have seen is attained by $RFL$ in expectation. Intuitively, the reason why we will be able to go from $n$ to $k^*$ is because in the congestion model we have the guarantee that $OPT$ will not have more than $k^*$ requests allocated to any facility so if we take one facility opened by $OPT$ and consider the requests allotted to it, we can expect to somehow get a competitive ratio of $\frac{\log k^*}{\log \log k^*}$. After that we will just have to show that the competitive ratio holds true even when there are multiple facilities opened by $OPT$.

We shall use the algorithm $RFL$ as a black box in order to obtain this modified algorithm, which we shall name $MRFL$.

\begin{algorithm}
  \caption{Modified RFL (MRFL)}\label{MRFL}
  \begin{algorithmic}[1]
    \Procedure{$MRFL$}{$x_i$}\Comment{Input request is $x_i$}
      \State On input request $x_i$, run algorithm $RFL$ on the current state for the input $x_i$ and perform the same action.
      \If{Any facility, $F$, has $k^*$ requests allocated to it}
        \State ``Close" facility $F$
        \State Open a new facility at the same location as $F$
    \EndIf
    \EndProcedure
  \end{algorithmic}
\end{algorithm}
We shall now prove \cref{thm1}-\cref{enu:cong_ub}, in particular for the Algorithm $MRFL$.
\begin{proof}[Proof of {\cref{thm1}-\cref{enu:cong_ub}}]
Firstly, by closing a facility, we mean that no more requests will be allocated to the facility anymore by the algorithm. No change in the model is assumed since the facility still exists: just that the algorithm will not allocate any more requests to it. Since another facility is opened right after a facility closes, at the same exact location, for future queries the algorithm will have this new facility available. Therefore the actions taken by $RFL$ in the model without congestion (For e.g., allocating a request to a facility which has now become ``closed'') can always be performed by $MRFL$ in the model with congestion (Since the ``closed'' facility has another facility open at the exact same location). Also apart from the extra facilities opened by $MRFL$ in step $5$, the actions taken by $RFL$ and $MRFL$ are identical. This means that the distance cost paid by the algorithm $RFL$ on a model without congestion is the same as the distance cost paid by the algorithm $MRFL$ on a model with congestion whenever the sequence of input requests is the same. We cannot say the same for the facility opening cost though, since the algorithm $MRFL$ opens excess facilities in step $5$. However if we separate the facilities opened by $MRFL$ into two categories : facilities opened by $MRFL$ in step $2$ and facilities opened by $MRFL$ in step $5$, we will notice that the number of facilities opened by $MRFL$ in step $2$ is equal to the number of facilities that $RFL$ would have opened in the congestion free model for the same sequence of input requests. Visually, the total number of points in the metric space where a facility is opened is same for $MRFL$ and $RFL$, just that $MRFL$ might have multiple facilities opened at a few points in the metric space. Also the congestion cost is something that only $MRFL$ has to pay and $RFL$ does not need to pay since $RFL$ runs on a congestion free model.

Based on this observation, we shall now try to divide the total cost paid by $MRFL$ into two parts. Let $C$ be the total cost paid by $MRFL$ over the entire sequence of input requests. Let $C_1$ be the total distance cost paid by $MRFL$ plus the total facility opening cost paid by $MRFL$ for facilities opened in step $2$ of the algorithm. Let $C_2$ be the total congestion cost paid by $MRFL$ plus the total facility opening cost paid over all facilities opened in step $5$ of the algorithm.

Clearly $C=C_1+C_2$ and hence we can write the competitive ratio of algorithm $MRFL$ as
$$\mu=\frac{C}{C_{OPT}}=\frac{C_1+C_2}{C_{OPT}}\;.$$
where $C_{OPT}$ is the cost paid by an optimal offline algorithm on the same sequence of input requests.

\begin{claim}\label{c2}
If the total number of input requests is $n$ then the cost $C_2$ incurred by algorithm $MRFL$ is at most $\frac{n}{k^*}f\left(1+\frac{g(2)}{g(2)-2}\right)\;.$
\end{claim}
\begin{proof}
In order to prove this, let use see which costs are counted in $C_2$. The entire congestion cost is counted in $C_2$ and also the facility opening cost for the facilities opened in step $5$ of the algorithm are counted in $C_2$. The two different costs in $C_2$ are actual costs of $MRFL$ for a particular run of the algorithm. Let us deal with these two costs individually.

In order to bound the facility opening cost in $C_2$, we need one key observation : every time a new facility is opened by $MRFL$ in step $5$, the algorithm has closed a facility. In order for a facility to be closed, there must have been $k^*$ requests allotted to it. Therefore every time a facility is opened in step $5$, it is because of $k^*$ distinct requests that the previous facility (which has now been closed) was allocated to. Hence, at most $\frac{n}{k^*}$ facilities can be opened in step $5$. Therefore the facility opening cost contributing towards $C_2$ is at most $\frac{n}{k^*}f$.

Now we would like to bound the total congestion cost paid by $MRFL$. In order to do this, we shall use the fact that the congestion function, $g$ is a convex non-decreasing function. Let the algorithm $MRFL$ open $k$ facilities and let there be $a_i$ requests allotted to the $i$-th facility for $i=1,2,...,k$. Then the total congestion cost paid by $MRFL$ is $\sum\limits_{i=1}^k g(a_i)$.

We have $\sum\limits_{i=1}^k a_i=n$, which is fixed and we want to see the maximum value attained by $\sum\limits_{i=1}^k g(a_i)$. This maximum is attained when $k$ is least and the $a_i$'s are the largest- that is, when there are the least number of facilities and the facilities have a large number of requests allotted to them. So we could have had a congestion cost as large as $g(n)$. However the way that this algorithm $MRFL$ is designed, we have a guarantee that no facility has more than $k^*$ requests allotted to it. Hence the maximum congestion cost by $MRFL$ is at most $\frac{n}{k^*}g(k^*)=\frac{n}{k^*}\frac{f}{g(2)-2}g(2)\;.$

Adding up both these types of costs we get that 
$$C_2\leq \frac{n}{k^*}f+\frac{n}{k^*}\frac{f}{g(2)-2}g(2)=\frac{n}{k^*}f\left(1+\frac{g(2)}{g(2)-2}\right)\;.$$
\end{proof}
At this point, using \cref{opt-bound}, we can say the OPT must open at least $\frac{n}{k^*}$ facilities and hence $C_{OPT}\geq \frac{n}{k^*}f$. Also using \cref{c2}, we  know that $C_2\leq \frac{n}{k^*}f\left(1+\frac{g(2)}{g(2)-2}\right)$. Therefore,
\begin{equation}\label{C2-eq}
    \frac{C_2}{C_{OPT}}\leq 1+\frac{g(2)}{g(2)-2}\;.
\end{equation}
Hence we can focus our attention on $\frac{C1}{C_{OPT}}$, since $\frac{C_2}{C_{OPT}}$ is a constant. Let us first try to understand what these terms represent. $C_1$ represents the total distance cost paid by the algorithm $MRFL$ plus the facility opening cost for facilities opened in step $2$ of the \cref{MRFL}. We notice that the distance cost paid by $MRFL$ over a sequence of input requests is equal to the distance cost paid by $RFL$ over the same sequence of input requests in the model without congestion. Also, the number of facilities opened by $RFL$ is the same as the number of facilities opened by $MRFL$ in step $2$ by construction. Hence $C_1$, in expectation, is the cost paid by $RFL$ on the same sequence of input requests in the model without congestion.

However the denominator $C_{OPT}$ is the cost paid by $OPT$ for the same sequence of input requests in the model with congestion. One thing to note here is that if the denominator was the cost paid by $OPT$ without congestion, we would have been stuck because then this ratio would have been $\mu(n):=\frac{\log n}{\log \log n}$, which is much larger than the competitive ratio of $\frac{\log k^*}{\log \log k^*}$ we are trying to achieve.

At this point we would have ideally liked to say that since $OPT$ allocates at most $k^*$ requests to any facility, let us take any particular facility $c^*$ opened by $OPT$. If we looked only at the requests allotted to $c^*$ by $OPT$, the cost paid by algorithm $RFL$ is at most $\mu(k^*)$ times the cost paid by $OPT$ had the input been only these requests allotted to $c^*$ by $OPT$. However the input actually comprises of several such clusters, where each cluster is the set of requests that were allocated to the same facility by $OPT$. While it is true that the competitive ratio would have been $\mu(k^*)$ if the input was only one such cluster, it is not clear if that will be the case for the entire input sequence. This is because while $OPT$'s cost over the entire input sequence is exactly the sum of costs of the clusters, $RFL$ might end up paying more on the entire input when compared to $RFL$'s costs on the clusters individually (had the input been just the cluster), summed up.

More formally, let us view the entire input sequence of length $n$ as clusters of input sequences where each such input sequence corresponds to the requests allotted to a given facility. That is $n=\sum_j a_j$, where $a_j$ is the number of requests allotted to the $j$-th facility. Also, we know that $a_j\leq k^*$ for all $j$. Additionally had the input been just the $a_j$ requests, $\frac{C_1}{C_{OPT}}$ would have been less than $\mu(k^*)$. Now that the various clusters are given as input together, the cost paid by $OPT$ is exactly the sum of the costs of the clusters separately. However for the online algorithm $RFL$ the cost paid on the entire input might be more than the sum of the costs over the clusters (assuming that the input was just the cluster). Our initial guess was that this cost $C_1$ over the entire sequence might still be less than some constant times the sum of the costs of the clusters individually. However we have been unable to prove this.

We observe that the denominator has $OPT$'s cost in the model with congestion. So $OPT$ would have at most $k^*$ requests allotted to any facility (Using \cref{opt-bound}). Let us now concentrate on any facility opened by $OPT$ and look at the requests that are allocated to the facility. So let us write $S=S_1\cup S_2\cup ...\cup S_k$, where $S$ is the entire collection of input requests and $S_i$ is the set of input requests that were allotted to a particular facility and $k$ is the total number of facilities opened by $OPT$. Let $C_{S_i}$ represent the cost paid by $OPT$ in the model with congestion if the input was just $S_i$ instead of $S$ and let $\Tilde{C}_{S_i}$ represent the cost paid by $OPT$ if the input was $S_i$ in the model without congestion. Similarly let $C_{1_{S_i}}$ be the cost paid by $RFL$ (recall that for this we are only looking at the model without congestion) over the requests in $S_i$ but when the input is the entire input $S$.
\begin{align*}
    \frac{C_1}{C_{OPT}}&=\frac{C_1}{\sum_i C_{S_i}}&&\\
    &\leq\frac{C_1}{\sum_i \Tilde{C}_{S_i}}&&\text{Since $C_{S_i}\geq \Tilde{C}_{S_i}$ by definition}\\
    &=\frac{\sum_i C_{1_{S_i}}}{\sum_i \Tilde{C}_{S_i}}&&\text{By definition of $C_{1_{S_i}}$}\\
    &\leq \max_i\left\{\frac{C_{1_{S_i}}}{\Tilde{C}_{S_i}}\right\}&&
\end{align*}
Here the denominator, $\Tilde{C}_{S_i}$ is the cost paid by $OPT$ if the input was $S_i$ in the model without congestion. However the numerator is $C_{1_{S_i}}$, which is the cost paid by $RFL$ in the model without congestion when the input is the entire input $S$. Had the numerator been the cost paid by $RFL$ in the model without congestion when the input is just $S_i$, we would have stated that this is exactly $\mu(|S_i|)$. Also since $|S_i|\leq k^*$ for all $i$, we would have obtained $\frac{C_1}{C_{OPT}}\leq \mu(k^*)=\mathcal{O}\left(\frac{\log k^*}{\log \log k^*}\right)$ for $RFL$.

Now the only problem that we have here is the fact that $C_{1_{S_i}}$ is the cost paid by $RFL$ in the model without congestion when the input is the entire input $S$. This basically means that now $RFL$ might have extra facilities already open when the requests of $S_i$ start coming in. Also new facilities might open between any two requests of $S_i$ since other requests from $S$ might be interleaved in between requests from $S_i$. The analysis of $RFL$ in \cite{Fotakis2008} actually goes through\footnote{Also any algorithm that satisfies this property, could have been used as a black box instead of $RFL$} if there were open facilities available before the first request comes and also if new facilities appeared in between two requests (because of facilities opening due to interleaved requests from other clusters). For the sake of completion, we add a thorough analysis of this in the Appendix in \cref{RFLC}. This completes the proof that the competitive ratio of $MRFL$ is $\frac{\log k^*}{\log \log k^*}$.
\end{proof}

\newpage
\bibliographystyle{plainurl}
\bibliography{refs}

\appendix

\section{Facility Location with bounded-away requests}\label{sec:bounded}
The sequence of input requests that gave us the lower bound in \cref{lbwc} are a modification of Fotakis's analysis in \cite{Fotakis2008}. These requests were such that the distance between consecutive pairs of requests kept getting smaller and smaller. In some real life applications of the online facility location problem, the requests can be such that no two requests are arbitrarily close to one another and one may assume that the distance between any two requests is bounded away. Hence, one might ask whether such a lower bound is possible when the distance between pairs of requests is not arbitrarily small.

In this section we shall prove that if the distance between pairs of requests is bounded away with respect to $f$, then such a lower bound is not possible. More precisely, we show that if the distance between any two requests is at least $d$ and $\nicefrac{f}{d}$ is not arbitrarily large then $RFL$ attains a competitive ratio of $\mathcal{O}\left(\frac{\log\nicefrac{f}{d}}{\log\log\nicefrac{f}{d}}\right)$ in expectation. Notice that if $\nicefrac{f}{d}$ is a constant independent of $n$ (the number of requests), this tells us that Meyerson's algorithm achieves a \emph{constant} competitive ratio.  This follows directly from the work of Meyerson, we make this formal by analysing this explicitly.

\begin{algorithm}
  \caption{Randomized Facility Location (RFL)}\label{RFL}
  \begin{algorithmic}[1]
    \Procedure{$RFL$}{$x_i$}\Comment{Input request is $x_i$}
      \State $F_i\gets \texttt{Nearest facility to }x_i$
      \State $p_i\gets \min\{1,\frac{d(x_i,F_i)}{f}\}$
      \State With probability $p_i$ open a new facility at $x_i$ and allocate $x_i$ to it
      \State With probability $1-p_i$, assign $x_i$ to $F_i$
    \EndProcedure
  \end{algorithmic}
\end{algorithm}

\begin{theorem}
If any two requests are at least $d$ distance apart, then $RFL$ attains a competitive ratio of $\mathcal{O}(\frac{\log f/d}{\log \log f/d})$.
\end{theorem}
\begin{proof}
The proof will be the exact same proof done by Fotakis, just using the additional information that vertices cannot be arbitrarily close. So, once again let $c^*$ be a facility opened by $OPT$. For those requests allocated to $c^*$, $OPT$ pays $f$ cost for opening facility at $c^*$ and $\sum_{i=0}^n d(x_i,c^*)$ total distance cost.

Now let us draw a ball at $c^*$ of radius $\nicefrac{d}{4}$ (We shall refer to this ball as $B(c^*,\nicefrac{d}{4})$). Since any two requests are at least $d$ distance apart, this ball can at most contain $1$ requests. This is the additional property that we did not have earlier. Also there cannot be any request more than $f$ distance away from $c^*$ that $OPT$ has assigned to $c^*$. So all the remaining requests are inside $B(c^*,f)\setminus B(c^*,\frac{d}{4})$, which we shall now divide into concentric annuli, as was done originally. 

\begin{figure}[!ht]
\centering
\fbox{
\begin{adjustbox}{width=0.7\textwidth}
\begin{tikzpicture}[line cap=round,line join=round,>=triangle 45,x=1cm,y=1cm]
\clip(-6.1,-6.1) rectangle (6.1,6.1);
\draw [line width=2pt] (0,0) circle (0.75cm);
\draw [line width=2pt,dash pattern=on 10pt off 10pt] (0,0) circle (1.9cm);
\draw [line width=2pt,dash pattern=on 10pt off 10pt] (0,0) circle (3.2cm);
\draw [line width=2pt] (0,0) circle (5cm);
\draw (-0.009172377727839933,0.3462314706029235) node[anchor=north west] {\Large$c^*$};
\draw (0.4,-0.4) node[anchor=north west] {\Large$\frac{d}{4}$};
\draw (3.8,-3) node[anchor=north west] {\Large$f$};
\begin{scriptsize}
\draw [fill=black] (0,0) circle (1.5pt);
\end{scriptsize}
\end{tikzpicture}
\end{adjustbox}}
\end{figure}

We shall have the following annuli (for some $m$ which we shall define later) :
\begin{itemize}
    \item $B(c^*,m\frac{d}{4})\setminus B(c^*,\frac{d}{4})$
    \item $B(c^*,m^2\frac{d}{4})\setminus B(c^*,m\frac{d}{4})$
    
    \hspace{1.5cm} \vdots
    \item $B(c^*,m^h\frac{d}{4})\setminus B(c^*,m^{h-1}\frac{d}{4})$
\end{itemize}

Where $h$ is such that $m^h\frac{d}{4}\geq f$, such that the last annulus has radius larger than $f$. This will guarantee that any possible request allocated to $c^*$ must necessarily lie in one of the annuli (and at most one request can lie in the innermost ball). Now we shall do the exact same analysis as \cite{Fotakis2008} on each of these annuli. Given any annulus, let us look at the requests in the annulus and ask how much $RFL$ pays for these requests in expectation. Using \cref{prop1} (We can use this because Online Facility Location can be considered as a special case of Weighted Online Facility Location where each request has unit weight and then WRFL is exactly same as RFL), we know that $RFL$ pays a distance cost $\leq f$ before a facility opens inside the annulus and then $RFL$ pays an additional cost of $f$ for opening the first facility inside the annulus.

After this, we can assume that for all the other requests in the annulus, there is an active facility in the annulus when the request arrives. Now if the request is in the $i$-th annuli, its distance from $c^*$ is at least $m^{i-1}\frac{d}{4}$ and by triangle inequality, no matter where $RFL$ opens the facility in the annulus, the request can be at most at a distance of $2m^i\frac{d}{4}$ from $RFL$'s open facility. Now by \cref{prop2} (Again, we can use this because it is a special version where the weight units are $1$ and $RFL$ is same as $WRFL$), $RFL$ pays in expectation a cost less than $4m^i\frac{d}{4}$ on the request, in comparison to $OPT$ paying $m^{i-1}\frac{d}{4}$ at least. This means that for each request $RFL$ pays at most $4m$ times the distance cost paid by $OPT$ in expectation.

However, like we calculated, $RFL$ also paid a cost of $2f$ for each annulus, on opening the first facility and also on distance costs before opening the first facility in the annuli. Since there are $h$ annuli, $RFL$ pays a cost of $2f.h$ in total, over all annuli. So if $OPT$ pays a distance cost of $C$ over the requests allotted to $c^*$, $OPT$ pays a total of $f+C$ for these requests. On the other hand $RFL$ pays less than $2fh+4mC$ in expectation. Also, there may or may not be a request in the ball $B(c^*,\frac{d}{4})$, which we have not yet accounted for in $RFL$'s analysis. However for one request $RFL$ can at most incur a cost of $f$.

Therefore we have $C_{OPT}=f+C$, while $C_{RFL}\leq 2fh+4mC+f=f(2h+1)+C(4m)$. So competitive ratio will be $\mu=\max \{2h+1,4m\}$. However $h$ and $m$ cannot be arbitrarily small since we must have $m^h d/4\geq f$ as we have seen earlier. Therefore $m^h\geq \nicefrac{4f}{d}$. Choosing $m=h=\frac{\log \nicefrac{4f}{d}}{\log \log \nicefrac{4f}{d}}$, we can satisfy the condition and obtain a competitive ratio of $\mu=\mathcal{O}\left(\frac{\log \nicefrac{4f}{d}}{\log \log \nicefrac{4f}{d}}\right)=\mathcal{O}\left(\frac{\log \nicefrac{f}{d}}{\log \log \nicefrac{f}{d}}\right)$, which is independent of $n$.
\end{proof}

\section{Missing Proofs}\label{sec:misproof}

\begin{proof}[Proof of \cref{prop1}]
Let $X_i=\{(x_i,w_i),(x_{i+1}w_{i+1}),\dots,(x_m,w_m)\}\;.$

Let $C_i$ be the expected cost paid by $WRFL$ before a facility opens at any request in $X_i$, when the input is $X_i$. Note that now, we are interested in $C_1$, which is the expected cost paid by $WRFL$ before a facility opens in $X$. From \cref{WRFL}, we recall that $WRFL$ opens a new facility at the request $x_i$ with probability $p_i$ and with the remaining probability it assigns $x_i$ to the nearest open facility.
\par
One can see that $C_1= p_1\cdot 0+(1-p_1)(p_1\cdot f+C_2)$ since with probability $p_1$ a facility is opened at $x_1$ itself and no cost is incurred before the first facility is opened but with probability $(1-p_1)$, no facility is opened and a cost of $p_1\cdot f=d(x_1,F)\cdot w_1$ is incurred. However since we have not yet opened any facility, in expectation $C_2$ additional cost will be incurred before a facility is opened - because of how $C_2$ is defined.

We can use the same argument for any $C_i$ now. That is, when $X_i$ is given as input, with probability $p_i$ a facility is opened at $x_i$ itself and no cost is incurred before the first facility is opened and with probability $(1-p_i)$, no facility is opened and a cost of $p_i\cdot f$ is incurred
$$C_i= (1-p_i)(p_i\cdot f+C_{i+1})$$

Now notice that $C_m$ is the expected cost paid by $WRFL$ before a facility opens when input is $\{(x_m,w_m)\}$. Therefore $C_m= p_m\cdot f\cdot (1-p_m)$ since with probability $(1-p_m)$, facility is not opened and then the cost paid is $p_m\cdot f$. This means that $C_m\leq f$.
\par
Assuming $C_{i+1}\leq f$ for any $i$ gives us 
$$C_i\leq (1-p_i)(p_i\cdot f+f)=f(1-p_i)(1+p_i)=f(1-p_i^2)\leq f\;.$$
$\therefore{} C_1\leq f\;.$
\end{proof}

\begin{proof}[Proof of \cref{prop2}]
Since the distance of $x_i$ to the nearest open facility is $d(x_i,F)$, $WRFL$ would open a new facility with probability $\leq \frac{d(x_i,F)\cdot w_i}{f}$ and pay a cost $f$ and with remaining probability (of at most $1$) it would pay a cost of $d(x_i,F)\cdot w_i$. Therefore expected cost paid for the request $x_i$ is at most
$$\frac{d(x_i,F)\cdot w_i}{f}\times f+1\times d(x_i,F)\cdot w_i=2d(x_i,F)\cdot w_i\;.$$
\end{proof}

\begin{claim}\label{clm1}
Let $WRFL$, in expectation, incur a cost of $C_{WRFL}(S_{c^*})$ over the requests in $S_{c^*}$ and let $OPT$ incur a cost of $C_{OPT}(S_{c^*})$ over the requests in $S_{c^*}$. Then the competitive ratio of $WRFL$ is at most $\max\limits_{c^*}\frac{C_{WRFL}(S_{c^*})}{C_{OPT}(S_{c^*})}$, in expectation.
\end{claim}
\begin{proof}
The competitive ratio of $WRFL$ is by definition $\frac{\sum_{c^*}C_{WRFL}(S_{c^*})}{\sum_{c^*}C_{OPT}(S_{c^*})}\leq \max\limits_{c^*}\frac{C_{WRFL}(S_{c^*})}{C_{OPT}(S_{c^*})}$.

\end{proof}

\section{Analysis of RFL}
Here we shall present the analysis of $RFL$ to show that it attains a competitive ratio of $\frac{\log n}{\log \log n}$ in expectation. Throughout the analysis we shall maintain the fact that there might be open facilities even before the first request arrives. Also new facilities might open in between two requests. This is a necessary property of the analysis that we had relied on while analyzing $MRFL$.

This analysis was done by Fotakis in \cite{Fotakis2008}.
\begin{theorem}\label{RFLC}
There are two parts to this theorem :
\begin{itemize}
    \item In the congestion free model, if $OPT$ has $n$ requests allotted to a facility, then the cost ratio of $RFL$ over these $n$ requests is $\frac{\log n}{\log \log n}$ in expectation.
    \item If $OPT$ assigns at most $n$ requests to a facility in the congestion free model then the competitive ratio of $RFL$ is $\frac{\log n}{\log \log n}$ in expectation.
\end{itemize}
\end{theorem}

Note that part $2$ of the Theorem follows from part $1$ of the Theorem, which is what we shall focus our attention on.
\begin{proof}[Proof of Part $1$ of {\cref{RFLC}}]
Let $c^*$ be the facility opened by $OPT$ that we are considering. Let $S$ denote the total distance cost of requests that are allocated to $c^*$ and let $A$ denote the average distance of these requests from $c^*$.
\begin{align*}
    S&=\sum_i d(x_i,c^*)&&\texttt{Where $x_i$'s are the requests assigned to $c^*$\;.}\\
    A&=\frac{S}{n}&&\texttt{Where $n$ is the number of requests assigned to $c^*$\;.}
\end{align*}
Therefore, the cost paid by $OPT$ over these requests assigned to $c^*$ is $f+S$.

Let us now split up the metric space containing the requests into concentric circles centered at $c^*$ and a ball centered at $c^*$. In order to do this, let $m$ be a constant that we shall determine later.
\begin{itemize}
    \item $B(c^*,A)$
    \item $B(c^*,mA)\setminus B(c^*,A)$
    \item $B(c^*,m^2A)\setminus B(c^*,mA)$
    \item $B(c^*,m^3A)\setminus B(c^*,m^2A)$
\end{itemize}
and so on. Please refer to \cref{fig1} for better visualisation.

\begin{figure}[ht]
\begin{framed}[0.9\textwidth]
\centering
\begin{adjustbox}{width=0.6\textwidth}
\begin{tikzpicture}[line cap=round,line join=round,>=triangle 45,x=1cm,y=1cm]
\clip(-11,-6) rectangle (0,5.4);
\draw [line width=2pt,dash pattern=on 5pt off 5pt] (-5.541568796284952,-0.28921560185752426) circle (1cm);
\draw [line width=2pt,dash pattern=on 10pt off 10pt] (-5.541568796284952,-0.28921560185752426) circle (2cm);
\draw [line width=2pt,dash pattern=on 10pt off 10pt] (-5.541568796284952,-0.28921560185752426) circle (4cm);
\draw [line width=2pt,dash pattern=on 10pt off 10pt] (-5.541568796284952,-0.28921560185752426) circle (8cm);
\draw (-5.653573439396889,-0.1660227356626117) node[anchor=north west] {\Large\textbf{$c^*$}};
\draw (-5.653573439396889,-1.1660227356626117) node[anchor=north west] {\LARGE\textbf{$A$}};
\draw (-5.499848624536694,-2.2460184303963507) node[anchor=north west] {\LARGE\textbf{$mA$}};
\draw (-5.538279828251743,-4.282872227293925) node[anchor=north west] {\LARGE\textbf{$m^2A$}};
\draw (-5.634357837539365,-8.260501811801453) node[anchor=north west] {\LARGE\textbf{$m^3A$}};
\begin{scriptsize}
\draw [fill=black] (-5.541568796284952,-0.28921560185752426) circle (1.5pt);
\draw [fill=ffqqqq] (-6.40297267187791,1.8687292498277395) circle (3.5pt);
\draw [fill=qqffqq] (-1.982970646329269,-0.38109160101458706) circle (3.5pt);
\draw [fill=qqffqq] (-8.672703441754239,-1.5557768240207575) circle (3.5pt);
\draw [fill=qqffqq] (-5.168557691769731,-3.327759618047015) circle (3.5pt);
\draw [fill=qqqqff] (-8.214775303972173,0.5347646745495458) circle (3.5pt);
\draw [fill=qqffqq] (-7.358648785510049,2.1872879543717856) circle (3.5pt);
\draw [fill=ffqqqq] (-8.831982794026263,1.1121523265356297) circle (3.5pt);
\end{scriptsize}
\end{tikzpicture}
\end{adjustbox}\caption{This diagram encapsulates the annuli and the ball at the centre, all concentric with centre at $c^*$. For one annulus, $RFL$ might not open a facility for the first few requests and use other facilities nearby (These requests are shown in Red). After this, on some request $RFL$ will open a facility on it (Shown in Blue). All future requests (Shown in green) on arrival will have the facility opened on Blue request (or other open facilities nearer to the request), where they may be allocated by $RFL$, if needed. Notice that the distance from the blue request to any green request in the annulus can be at most $2\cdot m^2A$\;.}\label{fig1}\end{framed}
\end{figure}

We will now try to estimate how much $RFL$ pays over all the requests that are allocated to $c^*$. We shall start with $B(c^*,A)$. Using \cref{prop1}, we know that $RFL$ pays a cost of $f$ in expectation before a facility is opened inside $B(c^*,A)$. Also in order to open the first facility, a facility opening cost of $f$ will have to be incurred resulting in an expected cost of $2f$. At this point, for all the future requests in $B(c^*,A)$, $RFL$ will have this facility available. So any request in $B(c^*,A)$ will have an open facility at a distance of at most $2A$ from it and hence the expected cost paid for any such request is at most $4A$, using \cref{prop2}. Also there can be at most $n$ requests allocated to $c^*$ by $OPT$ that are within $B(c^*,A)$ (Since there are $n$ requests allocated to $c^*$ by $OPT$).

Therefore the expected cost of the algorithm for the requests allocated to $c^*$ by $OPT$ and are inside $B(c^*,A)$ is at most
$$f+f+4A.n=2f+4S\;.$$

(Note that this entire reasoning holds true even if there are facilities open before the first request arrives and also if facilities appear in between two consecutive requests of $B(c^*,A)$.)

Next up let us consider any annulus : $B(c^*,m^iA)\setminus B(c^*,m^{i-1}A)$ and try to estimate the cost paid by $RFL$ for the requests in this annulus. Once again we shall use \cref{prop1} of $RFL$ to state that a cost of $f$ in expectation is paid at most before a facility is opened inside this annulus. After that a facility opening cost of $f$ is incurred for the first facility opened. At this point we will assume that there is an open facility for all future requests in $B(c^*,m^iA)\setminus B(c^*,m^{i-1}A)$. Hence on input $x_i$ in this annulus, we know that it is at least at a distance of $m^{i-1}A$ from $c^*$ and we also know that since $RFL$ has an open facility inside the annulus, $x_i$ has an open facility at a distance of at most $2.m^iA$ from it. Now we shall use \cref{prop2} to state that in expectation at most $4.m^iA$ cost is paid by $RFL$ on the request $x_i$. Therefore while $OPT$ pays a cost of at least $m^{i-1}A$ for $x_i$, $RFL$ in expectation might pay as much as $4.m^iA$ which is at most $4m$ times the distance cost that $OPT$ paid for $x_i$. We note that this is true for all requests $x_i$ opened by $RFL$ after it opens a facility in the annulus and also this is true for all annuli.

(Once again, we note that this entire reasoning holds true even if there are facilities open before the first request arrives and also if open facilities appear in between two consecutive requests of the annuli.)

Since the first annulus has outer radius $m^2A$, second annulus has outer radius $m^3A$ and so on, let us try to estimate the maximum number of annuli that we need to consider. Let us look at the $h$-th annulus which has an outer radius of $m^{h+1}A$. If $m^{h+1}A\geq S$, then certainly there are no requests outside this annulus. Now $m^{h+1}A\geq S\iff m^{h+1}\geq n$. Given an $m$, the smallest $h$ that satisfies $m^{h+1}\geq n$ would give us the number of annuli.

Now let us try to bound the cost that $RFL$ pays over all the requests on these annuli in expectation. Since there are at most $h$ annuli, the cost paid before a facility is opened and the first facility opening cost in the annuli together would add up to $2fh$. After the first facility is opened in each annuli, we used \cref{prop2} to state that in expectation $RFL$ pays a cost of $4m$ times the distance cost of the request. So summed up over all the requests in all the annuli, the expected cost paid by $RFL$ is at most $2fh+4mS$ since the sum of distance costs over all the requests is $S$ by definition. Adding the expected cost of the ball in the centre gives us $2fh+4mS+2f+4S=2f(h+1)+4S(m+1)$.

We note that the cost of $OPT$ was $f+S$, giving us a competitive ratio of $\mathcal{O}(\max\{h,m\})$ where $h$ and $m$ must satisfy $m^{h+1}\geq n$. This can be attained by choosing $m$ and $h$ such that $m=h=\Theta\left(\frac{\log n }{\log \log n}\right)$, which would give us $m^{h+1}\geq n$ and a competitive ratio of $\frac{\log n}{\log \log n}$ for $RFL$.

\end{proof}
\end{document}